\documentclass[10pt, onecolumn]{IEEEtran} 

\IEEEoverridecommandlockouts

\IEEEoverridecommandlockouts                              
\usepackage{enumerate}
\usepackage{amsmath,amssymb,algorithm,cite,cases,bm,float,graphicx,url}
\usepackage{pbox}

\usepackage{bm,float,url,color,array,tabu,pbox}

\usepackage{fullpage}

\usepackage{setspace}
\usepackage{amsthm}
\newtheorem*{theorem*}{Theorem}

\newtheorem*{lemma*}{Lemma}

\newtheorem{defn}{Definition}[section]
\newtheorem{appdefn}{Definition}[subsection]

\newtheorem{apppropo}{Proposition}[subsection]

\newtheorem{thm}{Theorem}[section]
\newtheorem{appthm}{Theorem}[subsection]

\newtheorem{applem}{Lemma}[subsection]

\newtheorem{cor}{Corollary}[section]

\usepackage{amsmath,amssymb,algorithm,cite,caption,cases,bm,float,graphicx,url,color,subfigure}
\usepackage[hidelinks]{hyperref}
\usepackage{comment}
\usepackage{thm-restate, enumerate }

\definecolor{darkred}{RGB}{150,0,0}
\definecolor{darkgreen}{RGB}{0,150,0}
\definecolor{darkblue}{RGB}{0,0,200}
\hypersetup{colorlinks=true, linkcolor=darkred, citecolor=darkgreen, urlcolor=darkblue}

\usepackage{graphicx}
\usepackage{caption}

\newcommand{\f}{\psi}
\newcommand{\Ic}{\mathcal{I}}

\newcommand{\xfun}{\upsilon^\mn}
\newcommand{\xfunx}{\xfun(\x;\g,\h)}

\newcommand{\Xfun}{\Upsilon^\mn}

\newcommand{\Xfunx}{\Xfun(\x;\G)}

\newcommand{\mn}{{(n)}}
\newcommand{\rn}{{\|\cdot\|}}

\newcommand{\xp}{\x_\phi(\g,\h)}
\newcommand{\xpn}{\x^\mn_\phi(\g,\h)}

\newcommand{\xPn}{\x^\mn_\Phi(\G)}

\newcommand{\rP}{\xrightarrow{P}}
\newcommand{\rD}{\rightarrow}

\newcommand{\rr}{f}

\newcommand{\Phn}{\Phi^\mn(\G)}
\newcommand{\phn}{\phi^\mn(\g,\h)}
\newcommand{\Ph}{\Phi(\G)}
\newcommand{\ph}{\phi(\g,\h)}

\newcommand{\NSE}{\operatorname{NSE}}

\newcommand{\az}{\alpha_*}

\newcommand{\D}{\mathbf{D}}

\newcommand{\Ws}{\mathcal{W}_*}

\newcommand{\loss}{\mathcal{L}}

\newcommand{\dom}{\operatorname{dom}}

\newcommand{\Ps}{\mathcal{P}}
\newcommand{\Rs}{\mathcal{R}}

\newcommand{\gw}{\omega}

\newcommand{\Exp}{\mathbb{E}}
\newcommand{\Ec}{\mathcal{E}}

\newcommand{\Hc}{\mathcal{H}}

\newcommand{\beq}{\begin{equation}}
\newcommand{\eeq}{\end{equation}}
\newcommand{\bea}{\begin{align}}
\newcommand{\eea}{{\end{align}}}

\newcommand{\vp}{\vspace{4pt}}

\newcommand{\Gb}{\mathbf{G}}
\newcommand{\G}{\mathbf{G}}

\newcommand{\A}{\mathbf{A}}

\newcommand{\w}{\mathbf{w}}

\newcommand{\x}{\mathbf{x}}
\newcommand{\ub}{\mathbf{u}}
\newcommand{\g}{\mathbf{g}}
\newcommand{\vb}{\mathbf{v}}
\newcommand{\bb}{\mathbf{b}}

\newcommand{\y}{\mathbf{y}}
\newcommand{\s}{\mathbf{s}}
\newcommand{\z}{\mathbf{z}}

\newcommand{\ab}{\mathbf{a}}

\newcommand{\h}{\mathbf{h}}

\newcommand{\Sc}{{\mathcal{S}}}

\newcommand{\Bc}{{\mathcal{B}}}

\newcommand{\Dc}{\mathcal{D}}
\newcommand{\Lc}{{\phi}}

\newcommand{\Kc}{\mathcal{K}}
\newcommand{\Nc}{\text{Null}}

\newcommand{\Nn}{\mathcal{N}}
\newcommand{\Cc}{\mathcal{C}}
\newcommand{\Gc}{{\Phi}}

\newcommand{\R}{\mathbb{R}}
\newcommand{\Pro}{{\mathbb{P}}}
\newcommand{\E}{{\mathbb{E}}}

\newcommand{\paf}{\pa f(\x_0)}

\newcommand{\la}{{\lambda}}

\newcommand{\eps}{\epsilon}

\newcommand{\Tc}{\mathcal{T}}
\newcommand{\pa}{\partial}

\newcommand{\vs}{\vspace}

\newcommand{\nn}{\nonumber}

\newcommand{\thename}{Gaussian min-max theorem}
\newcommand{\GMT}{GMT}

\begin{document}

\title{\LARGE{\bf{
The Gaussian Min-Max Theorem in the Presence of Convexity
}}\vs{1pt}}
%
\author{
Christos Thrampoulidis, Samet Oymak and Babak Hassibi\vs{2pt}\\Department of Electrical Engineering, Caltech, Pasadena 
}

%

\maketitle
\vspace{-10pt}
\begin{abstract} 
Gaussian comparison theorems are useful tools in probability theory; they are essential ingredients in the classical proofs of many results in empirical processes and extreme value theory. More recently, they have been  used extensively in the analysis of non-smooth optimization problems that arise in the recovery of structured signals from noisy linear observations. We refer to such problems as Primary Optimization (PO) problems.  A  prominent role in the study of  the (PO) problems is played by Gordon's \emph{Gaussian min-max theorem} (GMT) which provides probabilistic \emph{lower} bounds on the optimal cost via a simpler Auxiliary Optimization (AO) problem. Motivated by resent work of M. Stojnic, we show that under appropriate \emph{convexity assumptions} the (AO) problem allows one to \emph{tightly} bound both the optimal cost, as well as the norm of the solution of the (PO). As an application, we use our result to develop a general framework to tightly characterize the performance (e.g. squared-error) of a wide class of convex optimization algorithms used in the context of noisy signal recovery.
%
%
\end{abstract}

\section{Introduction}\label{sec:intro}

\vspace{-5pt}
\subsection{Motivation}
 Letting $\G\in\mathbb{R}^{{m}\times {n}}$, sets $\Sc_\x\subset\R^{n}$, $\Sc_\y\subset\R^{m}$ and a function $\psi(\cdot,\cdot):\Sc_\x\times\Sc_\y\rightarrow \R$, we define the following min-max optimization problem, which we shall henceforth refer to  as {\bf{\emph{Primary Optimization}} }(PO), 
\begin{align}
(PO)\qquad\qquad\Gc(\G,\Sc_\x,\Sc_\y,\psi) &:=  \min_{\x\in\Sc_\x}\max_{\y\in\Sc_\y}~ \y^T\G\x + \f({\x,\y}) 
.\label{eq:1}
\end{align}
A few popular instances of \eqref{eq:1} include the following:

\vp
i) \emph{Minimum signular value}: 
The minimum singular value of a matrix $\G$ can be written as
\begin{align}
\sigma_{\min}(\G):=\min_{\|\ab\|_2=1} \|\G\ab\|_2 = \min_{\|\ab\|_2=1} \max_{\|\bb\|_2= 1} \bb^T\G\ab = \Phi(\G,\Sc^{n-1},\Sc^{m-1},0),\label{eq:s_min}
\end{align}
where we have denoted $\Sc^{k-1}$ for  the unit sphere in $\R^k$.

\vp
ii) \emph{Minimum conic singular value}: 
Let $\Cc\subset\Sc^{n-1}$ be a subset of the unit sphere in $\R^n$. The  minimum conic singular value of a matrix $\G$ with respect to $\Cc$ is defined as (e.g. \cite{tropp2014convex})
$$
\sigma_{\min,\Cc}(\G):=\min_{\ab\in\Cc} \|\G\ab\|_2.
$$
This can be expressed in the format of the (PO) problem in \eqref{eq:1} as follows:
\begin{align}
\sigma_{\min,\Cc}(\G)= \min_{\ab\in\Cc} \max_{\|\bb\|_2=1} \bb^T\G\ab = \Phi(\G,\Cc,\Sc^{m-1},0).\label{eq:s_min2}
\end{align}

\vp
iii) \emph{LASSO}: 
The LASSO is a popular algorithm in the statistics literature \cite{TibLASSO,chen1998atomic} and is commonly used to recover a sparse signal $\x_0\in\R^n$ from a vector of noisy linear observations $\y=\G\x_0+\z\in\R^m$. It produces an estimate $\hat\x$ of $\x_0$ as the solution to the following convex optimization program\footnote{The formulation of the LASSO considered in \eqref{eq:intro_LASSO} is known as the ``square-root" \cite{Belloni} or $\ell_2$ \cite{OTH} LASSO. This is slightly different that the most well known formulation, $\min_{\x} (1/2)\|\G\x-\y\|_2^2 + \la \|\x\|_1$. Our theory applies to both variants,  as shown in Section \ref{sec:app}. }:
$$
\min_{\x} \|\y-\G\x\|_2 + \la \|\x\|_1,
$$
where $\la\geq 0$ is a regularization parameter. The first term $\|\y-\G\x\|_2$ attempts to fit $\x$ to the vector of observations $\y$ in a least-squares sense.
 The role of the regularizer $\|\x\|_1$ is to promote the sparsity of the solution. Observe that we can express the LASSO optimization problem as a (PO) problem as follows:
\begin{align}
\min_{\x} \|\G\x-\y\|_2 + \la \|\x\|_1 &= \min_{\x} \|\G(\x-\x_0)-\z\|_2 + \la \|\x\|_1 = \min_{\x}\max_{\|\ub\|_2\leq 1} \ub^T(\G(\x-\x_0)-\z) + \la \|\x\|_1\nn
\\&= \min_{\w}\max_{\|\ub\|_2\leq 1} \ub^T(\G\w-\z) + \la \|\x_0+\w\|_1 = \Phi(\G,\R^n,\Bc^{m},\psi),\label{eq:intro_LASSO}
\end{align}
where $\psi(\w,\ub) = -\ub^T\z + \la \|\x_0+\w\|_1$ and $\Bc^{m}$ is the unit ball in $\R^m$. In the above formulation, we have introduced the error vector $\w=\x-\x_0$ as the optimization variable. As explained later, our purpose will be to analyze the error performance of the LASSO, in which case this substitution becomes convenient.

\vp
iv) \emph{Estimating structured signals from noisy linear observations}: 
Consider the same setup as in the LASSO paradigm, i.e. estimating an unknown signal $\x_0$ from noisy linear observations $\y=\G\x_0 + \z$.  The LASSO can be viewed as a specific instance of a wide class of estimators that are obtained via solving:
\begin{align}
\min_{\x} \loss(\G\x-\y) + \la f(\x),\label{eq:intro_gen}
\end{align}
where, $\loss:\R^m\rightarrow\R$ is a \emph{convex} ``loss function" that penalizes the residual $\y-\A\x$. Typical examples include $\|\cdot\|_2,\|\cdot\|_2^2,\|\cdot\|_1$,etc.. On the other hand,  $f:\R^n\rightarrow\R$ is a regularizer function that aims to promote the specific structure of $\x_0$, e.g. $\ell_1$-norm if $\x_0$ is sparse, or nuclear-norm if $\x_0$ is a low-rank matrix in $\R^{\sqrt{n}\times\sqrt{n}}$. Letting $\loss^*:\R^m\rightarrow\R$ denote the conjugate function of $\loss(\cdot)$, and $\w=\x-\x_0$ denote the error vector, we can rewrite \eqref{eq:intro_gen} as a (PO) problem  as follows:
\begin{align}
\min_{\x} \loss(\G\x-\y) + \la f(\x) = \min_{\w}\sup_{\ub} \ub^T\A\w\underbrace{-\ub^T\z - \loss^*(\ub)+ \la\rr(\x_0+\w)}_{\psi(\w,\ub)} = \Phi(\G,\R^n,\R^m,\psi)\label{eq:intro_general_2}.
\end{align}
Here, we have used duality $\loss(\vb) = \sup_{\ub}\ub^T\vb-\loss^*(\ub)$; details are deferred to Section \ref{sec:app_gen}

\vspace{2pt}
\vp
In all the above instances consider $\G$ in \eqref{eq:1} being randomly drawn from a probability distribution. Then, $\Ps_{\x,\y} :=\y^T\G\x+\psi(\x,\y)$ can be viewed as a random process indexed on $\Sc_\x\times\Sc_\y$ and $\Phi(\G,\Sc_\x,\Sc_\y,\psi)$ is a random variable denoting its min-max value. Our goal is to understand the  distribution properties of $\Phi(\G,\Sc_\x,\Sc_\y,\psi)$. With this, we hope to be able to answer questions of the following flavor:
\begin{itemize}
\item[-] \emph{What is the minimum (conic) singular value of a random matrix?}
\item[-] \emph{How does the optimal cost of the LASSO,  or that of the general estimator in \eqref{eq:intro_general_2}, behave when the measurement matrix is random? More importantly, can we derive expressions for the reconstruction error $\|\hat{\x}-\x_0\|$ ? }
\end{itemize}
In particular, we restrict attention to the generic scenario in which the entries of $\G$ are i.i.d. standard normal $\Nn(0,1)$. Admittedly, this is a very special case of possible distributions of $\G$; in a large extent this is  driven by the fact that it allows us to rely on some remarkable properties that govern the Gaussian ensemble. However, it should be noted that many of relevant results obtained in random matrix theory for the Gaussian ensemble enjoy a \emph{universality} property, i.e. they actually hold for a wider class of probability distributions. Please refer to the discussion in Section \ref{sec:app_tight} and to \cite{korada2011applications,bayati2012universality,donoho2009observed} for instances of the universality property arising in the context of recovery of structured signals.

\subsection{The Gaussian min-max Theorem (GMT)}

Let $\g\in\R^m$ and $\h\in\R^n$ have entries i.i.d. $\Nn(0,1)$\footnote{Henceforth, we reserve notation $\G\in\R^{m\times n}$, $\g\in\R^m$ and $\h\in\R^n$ to denote r.v.s. of corresponding size and entries i.i.d. $\Nn(0,1)$.} and consider the following min-max optimization problem:
\begin{align}
(AO)\qquad\qquad\Lc(\g,\h,\Sc_\x,\Sc_y,\psi) &:=  \min_{\x\in\Sc_\x}\max_{\y\in\Sc_\y}~ \|\x\|_2\g^T\y + \|\y\|_2\h^T\x + \f({\x,\y}) 
.\label{eq:2}
\end{align}
$\Rs_{\x,\y}:=\|\x\|_2\g^T\y + \|\y\|_2\h^T\x + \f({\x,\y})$ is a gaussian random process indexed on $\Sc_\x\times\Sc_\y$ and $\phi(\g,\h,\Sc_\x,\Sc_\y,\psi)$ denotes its min-max value. 

We shall refer to the optimization problem in \eqref{eq:1} as the \emph{\bf{Auxiliary Optimization}} (AO) problem; it will be shown soon that the (AO) is closely related to the (PO).
At first glance it might seem that the two  min-max problems are somewhat unrelated to allow for this. Focusing on their objective functions, i.e. the Gaussian processes $\Ps_{\x,\y}$ and $\Rs_{\x,\y}$, it is seen that $\E_\G \Ps_{\x,\y} = \E_{\g,\h}\Rs_{\x,\y} = \psi(\x,\y)$. However, they don't even have the same variance, which would be a good indication that there might actually exist some relation\footnote{the intuition being that in the Gaussian ensemble, the first two moments are enough to capture the fundamentals of probabilistic behavior.}. Yet, Theorem \ref{lem:Gor} below, which is due to Gordon \cite{gorLem} and is known as the \emph{Gaussian min-max theorem} (GMT), shows that the two problems are indeed related. In particular, the theorem considers a slight modification of the (PO), as follows:
\begin{align}
\tilde\Gc(\G,g,\Sc_\x,\Sc_\y,\psi) &:=  \min_{\x\in\Sc_\x}\max_{\y\in\Sc_\y}~ \y^T\G\x + g\|\x\|_2\|\y\|_2 + \f({\x,\y}),\label{eq:3}
\end{align}
and, it relates that one to the (AO) problem. Observe that the modified Gaussian process involved in \eqref{eq:3} has now not only the same mean, but also the same variance as $\Rs_{\x,\y}$.
  For completeness, we include some background and a proof the GMT in Appendix \ref{sec:app_Gor}. Also, to simplify notation, we occasionally drop the arguments $\Sc_\x$, $\Sc_\y$ and $\psi$ from the defined variables in \eqref{eq:1}, \eqref{eq:2} and \eqref{eq:3}.
\begin{thm}[Gaussian min-max theorem (GMT) \cite{gorLem}\footnote{The version presented here is only a slight
modification   of the original result \cite[Lemma 3.1]{gorLem}.  In contrast to Theorem \ref{lem:Gor}, the original result in \cite[Lemma 3.1]{gorLem} assumes $\Sc_\x$ to be arbitrary (not necessarily compact) subset of $\R^m$, $\Sc_\y$ is restricted to be the unit sphere in $\R^n$ and $\psi(\cdot,\cdot)$ is only a function of $\x$.} ]
\label{lem:Gor}
Consider $\tilde\Phi(\G,g)$ and $\phi(\g,\h)$ as defined in \eqref{eq:3} and \eqref{eq:2}, respectively. 
Let $\G\in\mathbb{R}^{{m}\times {n}}$, $g\in\R$, $\g\in\mathbb{R}^{m}$ and $\h\in\mathbb{R}^{n}$ all have  entries i.i.d. $\Nn(0,1)$, $\Sc_\x\subset\R^{n}$, $\Sc_\y\subset\R^{m}$ be compact sets and $\psi(\cdot,\cdot):\Sc_\x\times\Sc_\y\rightarrow \R$ be continuous. 
Then, for any $c\in\mathbb{R}$,
$$\Pro(  \tilde\Gc(\G,g) < c ) \leq 
\Pro\left(  \Lc(\g,\h) \leq c  \right).$$
\end{thm}

In words, the lower tail probability of $\tilde\Phi(\G,g)$ is upper bounded by that of $\phi(\g,\h)$. This is a remarkable result, but not exactly serving our goal of understanding the properties of the optimal cost and optimal value of the (PO). Yet, with a simple trick we can show (cf. Section \ref{sec:proof}) that $\Pro( \Gc(\G) < c )\leq 2\Pro(  \tilde\Gc(\G,g) < c )$. Hence,
 \begin{align}\label{eq:Gor}
\Pro(  \Gc(\G) < c ) \leq 
2\Pro\left(  \Lc(\g,\h) \leq c  \right).
\end{align}
The message of the theorem translated to our setup is simple:
\begin{center}
{ \emph{  if $c$ is a high probability \textbf{lower} bound on the optimal cost of the (AO)\footnote{in the sense that $\Pro\left(\Lc(\g,\h)\leq c\right)$ is close to zero}, so it is for the optimal cost of the (PO). }}
\end{center}
 The most well-known application of this remarkable result is for the purpose of lower-bounding the minimum singular value of an $m\times n, m>n$ Gaussian matrix $\G$ (e.g. \cite{vershynin2010introduction}; see also \cite{Cha,tropp2014convex} for similar calculations regarding the minimum conic singular value). It is easily verified that the (AO) problem corresponding to \eqref{eq:s_min} is:
\begin{align}
\phi(\g,\h,\Sc^{n-1},\Sc^{m-1},0) = \|\g\|_2 - \|\h\|_2.
\end{align}
Standard concentration results (see Lemma \ref{lem:conc1}) show that the event $\|\g\|_2 - \|\h\|_2\leq \sqrt{m}-\sqrt{n}-t$ holds with probability at most $2\exp(-t^2/4)$ for all $t>0$. Hence, from \eqref{eq:Gor} and \eqref{eq:s_min}, we conclude that
\begin{align}
\Pro(\sigma_{\min}(\G)<\sqrt{m}-\sqrt{n}-t) \leq 4\exp(-t^2/4).\label{eq:Ver}
\end{align}
This example clearly suggests that the (AO) problem is in general simpler to analyze than the corresponding (PO). 

\subsection{Our Contribution}

A natural question that arises concerns the \emph{tightness} of the bounds obtained via Theorem \ref{lem:Gor}. To become explicit, suppose
 that $\Lc(\g,\h)$ concentrates around some constant $\mu$, in the sense that
for all $t>0$ , the events
$$
\left\{ \Lc(\g,\h) \leq \mu - t \right\} \quad \text{ and } \quad \left\{ \Lc(\g,\h) \geq \mu + t \right\},
$$
each occurs with \emph{low} probability.
Then, $\mu-t$ is a high-probability lower bound to $\Lc(\g,\h)$. This bound is also \emph{tight} in the sense that it is accompanied by a corresponding high-probability upper bound, namely $\mu+t$, whose value can be made arbitrarily close to the former.
Theorem \ref{lem:Gor} implies that $\mu-t$ is also a high-probability lower bound on $\Gc(\G)$. But, it gives \emph{no} information on how much $\Gc(\G)$ is allowed to deviate from this.

In  this note, we show that under additional \emph{convexity} assumptions on the sets $\Sc_\x$, $\Sc_\y$ and on the function $\psi(\cdot,\cdot),$ Theorem \ref{lem:Gor} is tight in the sense discussed above. We essentially prove that in the presence of convexity, the following counterpart of \eqref{eq:Gor} is true for all $c\in\R$:
 \begin{align}\label{eq:Gor2}
\Pro(  \Gc(\G) < c ) \leq 
2\Pro\left(  \Lc(\g,\h) \leq c  \right).
\end{align}
In words,
{\begin{center}
{ \emph{  if $c$ is a high probability \textbf{upper} bound on the optimal cost of the (AO), so it is for the optimal cost of the (PO). }}
\end{center}
}
\noindent Combining \eqref{eq:Gor} and \eqref{eq:Gor2} we conclude that under the appropriate convexity assumptions, for all $\mu\in\R$ and $t>0$:
\begin{align}\nn
\Pro\left( |\Ph - \mu| > t \right) \leq 2 \Pro\left( |\ph - \mu| > t \right).
\end{align}

Unfortunately, the constraint sets $\Sc_\x$, $\Sc_\y$ involved in \eqref{eq:s_min} and \eqref{eq:s_min2} are not convex. Hence, our result does not apply to upper bounding the minimum singular values\footnote{We remark, however, that  the lower bound obtained  in \eqref{eq:Ver} for $\sigma_{\min}(\G)$ is indeed \emph{asymptotically} tight in the regime $n/m\rightarrow(0,1), n\rightarrow\infty$ by the Bai-Yin's law \cite{bai1993limit}.}.
Yet, our result applies to the analysis of \eqref{eq:intro_general_2} and may have other potential applications. In \eqref{eq:intro_LASSO}, the principal objective is not characterizing the optimal cost of the optimization, but rather, its optimal minimizer $\hat\x$ and concluding about the achieved reconstruction error $\|\hat\x-\x_0\|_2$. With this serving as our motivation,  we prove 
 that in an asymptotic setting where the problem dimensions $m$ and $n$ grow to infinity, and under proper assumptions, 
 \begin{align}
 \|\x_\Phi(\G)\| \approx \|\x_\phi(\g,\h)\|\label{eq:Gor3}.
 \end{align}
Here,  $\x_\Phi(\G)$ and $\x_\phi(\g,\h)$ denote the optimal minimizers in the (PO) and (AO) problems, respectively.
In Section \ref{sec:app} we apply Theorem \ref{thm:main} to  derive a unifying framework for the asymptotically \emph{exact} error performance analysis of  non-smooth convex algorithms that can be cast in \eqref{eq:intro_general_2}. We point references to our accompanying series of works, which applies the framework to specific problem instances such as the LASSO or the regularized Least Absolute Deviations (LAD) algorithm.

This work is highly motivated and inspired by recent work of Stojnic \cite{stojnic2013meshes,stojnic2013upper,stojnic2013spherical};  we discuss connection and relevance to this  and other literature primarily in Section \ref{sec:rel} and throughout the main body of the paper. 
The rest of the paper is organized as follows. In Section \ref{sec:main} we state and discuss our main result Theorem \ref{thm:main}. The theorem consists of three statements corresponding to \eqref{eq:Gor}, \eqref{eq:Gor2} and \eqref{eq:Gor3}, respectively. The proof is included in Section \ref{sec:proof}. We use the machinery of Theorem \ref{thm:main} in Section \ref{sec:app} in order to evaluate the estimation performance of \eqref{eq:intro_general_2}.
 The paper concludes in Section \ref{sec:rel} with a discussion of the relevant work. Some of the technical proofs and further discussions are deferred to the Appendix.

\section{The convex Gaussian min-max Theorem}\label{sec:main}

\subsection{Some Notation}\label{sec:notation}

\begin{defn}[GMT admissible sequence] \label{def:ad}The sequence $\{ \G^\mn, \g^\mn, \h^\mn, \Sc_\x^\mn, \Sc_\y^\mn, \psi^\mn(\cdot,\cdot) \}_{n\in\mathbb{N}}$ indexed by n is said to be a \emph{GMT admissible sequence} if $m=m(n)$ and if for all $n\in\mathbb{N}$:
\begin{itemize}
\item $\G^\mn\in\R^{m\times n}, \h^\mn\in\R^n, \g^\mn\in\R^m$, 
\item $\Sc_\x^\mn\subset\R^n$ and $\Sc_\y^\mn\subset\R^m$ are compact sets, 
\item $\psi^\mn:\Sc_\x^\mn\times\Sc_\y^\mn\rightarrow\R$ is a continuous 
function.
\end{itemize}
Onwards, we will drop the $(n)$ superscript from $\G^\mn$,$\g^\mn$ and $\h^\mn$ to simplify notation somewhat. 
\end{defn}

\noindent A sequence $\{ \G^\mn, \g^\mn, \h^\mn, \Sc_\x^\mn, \Sc_\y^\mn, \psi^\mn(\cdot,\cdot) \}_{n\in\mathbb{N}}$ defines a sequence of min-max optimization problems as in \eqref{eq:1} and \eqref{eq:2}, i.e.,
\begin{subequations}
\begin{align}
\Phi^{(n)}(\G) &= \min_{\x\in\Sc_\x^\mn}\max_{\y\in\Sc_\y^\mn} \y^T\G\x + \psi^\mn(\x,\y),\label{eq:1a}\\
\phi^{(n)}(\g,\h) &= \min_{\x\in\Sc_\x^\mn}\max_{\y\in\Sc_\y^\mn} \|\x\|_2\g^T\y + \|\y\|_2\h^T\x+ \psi^\mn(\x,\y)\label{eq:2a}.
\end{align}
\end{subequations}
We correspondingly refer to those as the (PO) and (AO) problems.  Let the corresponding optimal  minimizers of those problems be denoted as $\x_\Phi^\mn(\G)$ and $\x_\phi^\mn(\g,\h)$, respectively. It is convenient for the statement of the theorem to define the sequence of (random) functions   $\xfun:\Sc_\x^\mn\rightarrow \R$:
  \begin{align}\label{eq:xfun}
  \xfunx = \max_{\y\in\Sc_\y^\mn} \|\x\|_2\g^T\y + \|\y\|_2\h^T\x+ \psi^\mn(\x,\y).
  \end{align}
  Clearly, $\Lc^\mn(\g,\h) = \min_{\x\in\Sc_\x^\mn}\xfunx$.
  


 For a sequence of random variables $\{\mathcal{X}^\mn\}_{n\in\mathbb{N}}$ and $c\in\R$, we use standard notation $\mathcal{X}^\mn\rP c$ to denote convergence in probability to $c$, i.e. for all $\eps>0$, $\lim_{n\rightarrow\infty}\Pro\left(|\mathcal{X}^{(n)} - c|>\eps\right)=0$. Similarly, for a deterministic sequence $\{x^\mn\}_{n\in\mathbb{N}}$ we write $x^\mn\rD c$ for $\lim_{n\rD\infty}x^\mn=c,~c\in\R$.

 

\subsection{Result}\label{sec:result}

We are now ready to state our main result in Theorem \ref{thm:main}.

\begin{thm}[convex Gaussian min-max theorem (cGMT)]\label{thm:main}
Let $\{ \G^\mn, \g^\mn, \h^\mn, \Sc_\x^\mn, \Sc_\y^\mn, \psi^\mn(\cdot,\cdot) \}_{n\in\mathbb{N}}$ be a GMT admissible sequence as in Definition \ref{def:ad}, for which the entries of $\G,\h$ and $\g$ are drawn i.i.d. $\Nn(0,1)$.
Consider $\Phi^\mn(\G)$ and $\phi^\mn(\g,\h)$ as in \eqref{eq:1a} and \eqref{eq:2a}, and, $\x_\Phi^{(n)}(\G)$ and $\x_\phi^{(n)}(\g,\h)$ optimal minimizers therein. 
The following three statements hold.
\begin{enumerate}[(i)]
\item For any $n\in\mathbb{N}$ and $c\in\mathbb{R}$,
\begin{align}\label{eq:primal}
\Pro\left(  \Gc^\mn(\G) < c\right) \leq 2\Pro\left(  \Lc^\mn(\g,\h) \leq c \right).
\end{align}
\item If $\Sc_\x^\mn$, $\Sc_\y^\mn$ are convex and  $\psi^\mn(\cdot,\cdot)$ is convex-concave on $\Sc_\x^\mn\times\Sc_\y^\mn$, then, for any $n\in\mathbb{N}$, $\mu\in\mathbb{R}$ and $t>0$: 
%
\begin{align}\label{eq:dual}
\Pro\left( |\Gc^\mn(\G) - \mu| > t \right) \leq 2 \Pro\left( |\Lc^\mn(\g,\h) - \mu| \geq t \right).
\end{align}
\item Assume the conditions of (ii) hold for all $n\in\mathbb{N}$. Let $\|\cdot\|$ denote some norm in $\R^n$ and recall \eqref{eq:xfun}. If, there exist deterministic values (independent of $m,n$) $d_*,\alpha_*$ and $\tau>0$ such that
\begin{enumerate}[(a)]
\item $\Lc^\mn(\g,\h) \rP d_*$,
\item $\| \x^\mn_\phi(\g,\h)\| \rP \alpha_*$,
\item with probability 1 in the limit $n\rightarrow\infty$: $ \{ \xfunx \geq \Lc^\mn(\g,\h) + \tau ( \| \x \| - \| \x^\mn_\phi(\g,\h) \| )^2, ~\forall \x\in\Sc_\x^\mn \}$,
\end{enumerate}
then,
\begin{align}\label{eq:norms}
\| \xPn \| \rP \alpha_*.
\end{align}

\end{enumerate}
\end{thm}


The probabilities are taken with respect to the randomness of $\G$, $\g$ and $\h$.
 A detailed discussion of the statements of the theorem follows in Section \ref{sec:discussion} below.





\subsection{Proof}\label{sec:proof}

Here, we only prove statements (i) and (ii) of Theorem \ref{thm:main}. The proofs are short and insightful. The proof of the third statement is not involved either, but essentially relies on the first two statements and requires some more space; thus, is deferred to Appendix \ref{sec:app_thm_proof}.
In what follows, we fix arbitrary $n\in\mathbb{N}$ and drop the superscript $(n)$ to simplify notation.


\vp\noindent\underline{Proof of \eqref{eq:primal}:}
 As discussed inequality \eqref{eq:primal} is an almost direct consequence of  Theorem \ref{lem:Gor}.
Yet we need to get rid of the term ``$g\|\x\|_2\|\y\|_2$" in \eqref{eq:Gor} in Gordon's Theorem \ref{lem:Gor}. The argument is rather simple but critical for the rest of the statements of Theorem \ref{thm:main}.
We will show that 
\begin{align}\label{eq:symGoal}
 \Pro\left(\Gc(\G)\leq c\right) \leq 2 \Pro\left(\tilde\Gc(\G,g)\geq c\right).
\end{align}
Once this is established, \eqref{eq:primal} follows directly after applying Theorem \ref{lem:Gor}.

To prove \eqref{eq:symGoal}, fix $\G$ and $g<0$ and denote
$$
f_1(\x,\y) = \y^T\G\x+ \psi({\x,\y}) \quad\text{ and }\quad f_2(\x,\y) = \y^T\G\x + g\|\x\|_2\|\y\|_2 + \psi({\x,\y}).
$$
 Clearly, $f_1(\x,\y)\geq f_2(\x,\y)$ for all $(\x,\y)\in\Sc_\x\times\Sc_\y.$
We may then write,
\begin{align}
\min_{\x\in\Sc_\x}\max_{\y\in\Sc_\y} f_1(\x,\y) = f_1(\x_1,\y_1) &\geq f_1(\x_1,\y) \text{ for all } \y\in\Sc_\y\nn\\
&\geq \max_{\y\in\Sc_\y} f_2(\x_1,\y) \geq \min_{\x\in\Sc_\x}\max_{\y\in\Sc_\y} f_2(\x,\y). \nn
\end{align}
This proves that
 $\Gc(\G)\geq \tilde\Gc(\G,g)$, when $g<0$. From this and from independence of $g$ and $\G$,
 for all $c\in\R$:
\begin{align*}
\Pro\left(\tilde\Gc(\G,g)\leq c ~|~ g<0\right) \geq \Pro\left(\Gc(\G)\leq c ~|~ g<0\right) = \Pro(\Gc(\G)\leq c).
\end{align*}
When combined with $g\sim\Nn(0,1)$, the above yields the desired inequality \eqref{eq:symGoal}:
\begin{align*}
\Pro\left(\tilde\Gc(\G,g)\leq c\right) &= \frac{1}{2}\Pro\left(\tilde\Gc(\G,g)\leq c ~|~ g>0\right) + \frac{1}{2}\Pro\left(\tilde\Gc(\G,g)\leq c ~|~ g<0\right) \nn 
\geq \frac{1}{2} \Pro(\Gc(\G)\leq c).
\end{align*}


\vp\noindent\underline{Proof of \eqref{eq:dual}:} 
The additional convexity assumptions imposed in  statement (ii) of the theorem are critical for the proof of \eqref{eq:dual}. By assumption, the sets $\Sc_\x,\Sc_\y$ are non-empty compact and convex. Furthermore, the function $\y^T\G\x + \f({\x,\y})$ is continuous, finite\footnote{A continuous function on a compact set is bounded from Weierstrass extreme value theorem.} and convex-concave on $\Sc_\x\times\Sc_\y$. Thus, we can apply the minimax result in \cite[Corollary 37.3.2]{Roc70} to exchange ``$\min$-$\max$" with a ``$\max$-$\min$" in \eqref{eq:1a}:
$$
\Gc(\G) = \max_{\y\in\Sc_\y}\min_{\x\in\Sc_\x} \y^T\G\x + \f({\x,\y}).
$$
It is convenient to rewrite the above as
\begin{align*}
-\Gc(\G) = \min_{\y\in\Sc_\y}\max_{\x\in\Sc_\x} -\y^T\G\x - \f({\x,\y}).
\end{align*}
Then, using the symmetry of $\G$, we have that for any $c\in\R$:
\begin{align*}
\Pro\left( -\Gc(\G) \leq c \right) = \Pro\left( \min_{\y\in\Sc_\y}\max_{\x\in\Sc_\x} \left\{\y^T\G\x - \f({\x,\y})\right\} \leq c \right).
\end{align*}
Thus, we may apply\footnote{Observe that the signs of  $\y^T\G\x$, $\g^T\y$ and $\h^T\x$ do not matter because of the assumed symmetry in the distributions of $\G,\g$ and $\h$.
} statement (i) of Theorem \ref{thm:main} (with the roles of $\x$ and $\y$ flipped):
\begin{align}
\Pro\left( -\Gc(\G)< c \right) &\leq 2\Pro\left(  \min_{\y\in\Sc_\y}\max_{\x\in\Sc_\x}\left\{  \|\y\|_2 \h^T\x + \|\x \|_2 \g^T\y  - \f({\x,\y})  \right\}\leq c  \right) \nn \\
&=2\Pro\left(  \min_{\y\in\Sc_\y}\max_{\x\in\Sc_\x}\left\{ - \|\y\|_2 \h^T\x -\|\x \|_2 \g^T\y - \f({\x,\y})  \right\}\leq c  \right),\label{eq:les?}
\end{align}
where the last equation follows because of the symmetry of $\g$ and $\h$. To continue, note that
$$
\min_{\y\in\Sc_\y}\max_{\x\in\Sc_\x}\left\{ - \|\y\|_2 \h^T\x -\|\x \|_2 \g^T\y - \f({\x,\y}) \right\} = - \max_{\y\in\Sc_\y}\min_{\x\in\Sc_\x}\left\{ \|\y\|_2 \h^T\x + \|\x \|_2 \g^T\y + \f({\x,\y})\right\},
$$
and further apply the minimax inequality \cite[Lemma 36.1]{Roc70} which requires that for all $\g,\h$
$$
\max_{\y\in\Sc_\y}\min_{\x\in\Sc_\x}\left\{ \|\x \|_2 \g^T\y  + \|\y\|_2 \h^T\x + \f({\x,\y}) \right\}\leq 
\min_{\x\in\Sc_\x}\max_{\y\in\Sc_\y}\left\{ \|\x \|_2 \g^T\y  + \|\y\|_2 \h^T\x + \f({\x,\y}) \right\} := \Lc(\g,\h).
$$
These, when combined with \eqref{eq:les?}, give
$\Pro\left( -\Gc(\G)< c \right) \leq 2\Pro\left( -\Lc(\g,\h)\leq c \right).$
Taking $c=-c_+$, proves that for all $c_+\in\R$:
\begin{align}\label{eq:dual22}
\Pro\left(  \Gc(\G) > c_+ \right) \leq 2\Pro\left(  \Lc(\g,\h) \geq c_+ \right). 
\end{align}
Fix any $\mu\in\R$ and $t>0$ and note that $\Pro(|x-\mu|>t) = \Pro(x<\mu-t)+ \Pro(x>\mu+t)$ for any r.v. $x\in\R$. Use this accordingly and apply \eqref{eq:primal} and \eqref{eq:dual22} for $c=\mu-t$ and $c_+=\mu+t$ to conclude with the desired.

\subsection{Discussion}\label{sec:discussion}


To enlighten notation, when clear from context, we  drop the superscript $(n)$ from the variables involved in Theorem \ref{thm:main}. Importantly, the first two are non-asymptotic in nature, i.e. they  hold true for all values of the problem dimensions $m$, $n$.

\vp
\subsubsection{Statement (i)}
Inequality \eqref{eq:primal} is essentially no different than what Theorem \ref{lem:Gor} states; if $c_-$ is a high probability lower bound for Gordon's optimization $\Lc(\g,\h)$, so it is for $\Gc(\G)$.
As already mentioned, in contrast to Theorem \ref{lem:Gor}, the minimax optimization in \eqref{eq:1a} does not include the term ``$g\|\x\|_2\|\y\|_2$". The ``price" paid for this, is the multiplicative factor of $2$ in \eqref{eq:primal}, when compared to \eqref{eq:Gor}. Note however that this factor does not affect the essence of the result since the scenarios of interest are those for which $\Pro(\Lc(\g,\h)\leq c_-)$ is close to zero. What is more, in most of the applications
where \GMT~has been proven to be useful, the optimization problem involved is in the form of \eqref{eq:1a} rather than that of \eqref{eq:3}. One reason behind this, is that under convexity assumptions on $\Sc_\x$, $\Sc_\y$ and $\psi(\cdot,\cdot)$ the minimax optimization in \eqref{eq:1a} is a \emph{convex} program, which is generally more likely to be encountered in applications compared to the always non-convex program in \eqref{eq:3}\footnote{ the component $g\|\x\|_2\|\y\|_2$ causes the min-max optimization in \eqref{eq:3} to be non-convex even when $\Sc_\x$,$\Sc_\y$ are convex and $\psi(\cdot,\cdot)$ convex-concave.}. Convexity, is also critical for establishing the second statement of the theorem, namely inequality \eqref{eq:dual}.

\vp
\subsubsection{Statement (ii)}
 The main contribution of Theorem \ref{thm:main} is   inequality \eqref{eq:dual}. This holds only after imposing appropriate convexity assumptions and provides a counterpart to \eqref{eq:primal} and \GMT.
%
Of course, \eqref{eq:dual} becomes interesting when $\mu$ is chosen so that $\ph$ concentrates around it. In this case, the probability in the right-hand side of \eqref{eq:dual} is vanishing, indicating that $\Ph$ concentrates around the same value.
In particular, we can apply \eqref{eq:dual} for $\mu = \E\Lc(\g,\h)$. It is shown in Lemma \ref{lem:lip} in the Appendix that $\Lc(\g,\h)$ is Lipschitz in $(\g,\h)$. It then follows from  the Gaussian concentration property of Lipschitz functions (see Proposition \ref{prop:lip}) that $\Lc(\g,\h)$ is normally concentrated around its mean $\Exp\Lc(\g,\h)$. Thus, we obtain Corollary \ref{cor:main} below.

\begin{cor}\label{cor:main} Consider the same setup as in Theorem \ref{thm:main} and let the assumptions of statement (ii) therein hold. Further, define $R_\x := \max_{\x\in\Sc_\x}\|\x\|_2$ and $R_\y := \max_{\y\in\Sc_\y}\|\y\|_2$.
 Then, for all $t>0$, 
\begin{align*}
\Pro\left(~|\Gc(\G) - \E\Lc(\g,\h) | > t ~\right)\leq 4\exp\left( -t^2/({4R_\x^2R_\y^2}) \right).
\end{align*} 
\end{cor}

Two remarks concerning the required convexity conditions  of the statement are in place. 

\underline{Remark 1}:
As seen in Section \ref{sec:proof}, the critical step involved in the proof of \eqref{eq:dual} is being able to flip the order of the min-max operation in the (PO) problem without changing its optimal cost. The convexity conditions as specified in the second statement of the theorem guarantee that this is possible. Note however that these conditions are only sufficient. In principle, it might be possible to flip the order of min-max under milder conditions in which case \eqref{eq:dual} would continue to hold.

\underline{Remark 2}: 
Note that flipping the min-max order, as in \cite[Corollary 37.3.2]{Roc70}, requires \emph{at least one} of the two sets $\Sc_\x$, $\Sc_\y$ to be compact. However, Theorem \ref{thm:main} asks for compactness of \emph{both} sets, the reason being to guarantee that \eqref{eq:primal} is applicable.
Overall, the compactness condition on the sets $\Sc_\x$ and $\Sc_\y$ is important. If we wish to apply Theorem \ref{thm:main} to problem instances, like \eqref{eq:intro_general_2}, in which the constraint sets appear unbounded, then we first need to show (if possible) that the set of min-max optima of \eqref{eq:1} is indeed compact. This would allow us to define proper sets $\Sc_\x$ and $\Sc_\y$. Please refer to Section \ref{sec:app_gen} and Appendix \ref{sec:tech}, for a concrete illustration of those ideas.




\vp
\subsubsection{Statement (iii)}
The first two statements of the theorem precisely characterize the relation between the optimal \emph{costs} $\Phn$ and $\phn$ of the two min-max optimizations in \eqref{eq:1a} and \eqref{eq:2a}: under appropriate convexity assumptions the tail distribution of $\Phn$ is bounded by two times the tail distribution of $\phn$. Hence, if $\phn$ has a ``well-behaved" limiting behavior, in the sense of converging to a deterministic limit $d_*$ as $n\rightarrow \infty$, then, $\Phn$ will also converge to the same limit.
Statement (iii) provides conditions under which a similar strong relation can be established between the optimal \emph{values} of the two seemingly unrelated optimization problems.

Let us briefly discuss the conditions of the statement. First, the same convexity assumptions as in statement (ii) should be present. Next, it is naturally required that as $n\rightarrow\infty$ both $\phn$ and $\|\xpn\|$ are ``well-behaved", i.e. they both converge to deterministic values say $d_*$ and $\alpha_*$. Here, it is important to remark that $\xpn$ denotes \emph{any} optimal minimizer in \eqref{eq:2a}. We do \emph{not} require that $\xpn$ is unique; there might be multiple such optima, but they all have norms that converge to $\alpha_*$. The last condition, guarantees that any other feasible $\x$ with norm that is far from the optimal $\alpha_*$ results in a  strictly positive increase (uniform over $n$) of the objective value. One sufficient (but not necessary) condition that is met in practice and satisfies the aforementioned is that the function $\xfun(\cdot;\g,\h)$ be \emph{strictly convex} with respect to the norm $\|\cdot\|$ in consideration. 

Satisfying the conditions of the third statement of the theorem requires thorough analysis of Gordon's optimization in \eqref{eq:2a}. In specific applications, $\Sc^\mn_\x,\Sc^\mn_\y$ and $\psi^\mn$ take explicit forms and may allow the simplification of the min-max optimization involved. The procedure is specific to different applications, but a somewhat general recipe that underlies the analysis in most cases can be derived. This is the subject of Appendix \ref{sec:recipe}.




\section{Application}\label{sec:app}

\subsection{Estimating structured signals from noisy linear observations}\label{sec:app_intro}

Consider
the task of estimating an unknown but \emph{structured} signal $\x_0\in\R^n$ from noisy linear observations $\y=\A\x_0+\z\in\R^m$, where $\A$ is the measurement matrix and $\z$ is the 
noise vector. $\x_0$ is structured in the sense that it actually lives in a manifold of lower dimension than the dimension $n$ of the ambient space. Typical examples of such signals include sparse and block-sparse signals, low-rank matrices, signals that are sparse over a dictionary (i.e. lies on a union of subspaces) and many more (see \cite{Cha}). 

In order to measure the fit of any vector $\x\in\R^n$ to the vector of observations $\y\in\R^m$ we introduce a loss function $\loss:\R^m\rightarrow\R$, which assigns a penalty $\loss(\y-\A\x)\geq0$ to the corresponding residual $\y-\A\x$. At the same time, 
in order to promote the particular structure of $\x_0$, we associate it with an appropriate structure-inducing function $\rr:\R^n\rightarrow\R$. For example, if $\x_0$ is a \emph{sparse} vector then $\rr$ can be the $\ell_1$-norm, while if $\x_0$ is a $\sqrt{n}\times\sqrt{n}$ \emph{low-rank} matrix then a popular choice for $\rr$ is the nuclear norm (see \cite{Cha} for more examples). 
With these we can obtain an estimate for $\x_0$ as the solution $\hat\x$ to the  following optimization procedure\footnote{the minimizer of \eqref{eq:genLASSO} need not be unique. Using a slight abuse of notation, let the operator $\arg\min$ return any one of those optimal values.}:
\begin{align}\label{eq:genLASSO}
\hat\x := \arg\min_{\x} \loss(\y-\A\x)+\la\rr(\x),
\end{align}
where $\la\geq 0$ is a regularizer parameter. If the functions $\loss(\cdot)$ and $f(\cdot)$ are both convex, then the optimization program in \eqref{eq:genLASSO} is convex, so it can be solved efficiently \cite{boyd2009convex}.
 Specific choices of the loss function $\loss$ and the regularizer $f$ give rise to different instances of the procedure in \eqref{eq:genLASSO}. In what follows, we enlist some of the most encountered such instances:
 \begin{itemize}
 \vp
 \item {Ordinary least-squares} ($\loss(\cdot) = (1/2)\|\cdot\|_2^2$, $f(\cdot) = 0$)
 \vp
 \item {LASSO} ($\loss(\cdot) = (1/2)\|\cdot\|_2^2$, $f(\cdot) = \|\cdot\|_1$). This particular choice amounts to arguably the most celebrated instance of \eqref{eq:genLASSO}. It is widely known in the statics literature as the LASSO and was proposed by Tibshirani in 1996 as a sparse recovery algorithm \cite{TibLASSO}. The  ``least-squares" nature of the loss function corresponds to a maximum likelihood estimator for the case when $\z$ is gaussian. 
 \vp
 \item {$\ell_2$-LASSO} ($\loss(\cdot) = \|\cdot\|_2$).  This algorithm is very similar in nature to the LASSO and is known in the literature as square-root LASSO \cite{Belloni} or $\ell_2$-LASSO \cite{OTH}. Indeed, Lagrange duality theory shows that there exist choice of the regularizer parameters of the two algorithms so that they become equivalent. However, there exists differences among them. As an example, it can be shown \cite{Belloni, OTH} that tuning of the regularizer parameter of the $\ell_2$-LASSO does not require knowledge of the standard deviation of the noise.
 \vp
 \item {Generalized-LASSO} ($\loss(\cdot) = \|\cdot\|_2^2$ or $\loss(\cdot) = \|\cdot\|_2$). This is a natural generalization of the LASSO and the $\ell_2$-LASSO that allows regularizers $f$ other than the $\ell_1$ norm. These include nonsmooth functions such as nuclear norm, $\ell_{1,2}$ norm and discrete total variation. We often refer to this as the Generalized LASSO.
 \vp
 \item {LAD} ($\loss(\cdot) = \|\cdot\|_1$). This uses a  ``least-absolute deviations" criterion as the loss function and is known as the regularized LAD estimator (\cite{wang2013,Allerton}) or $\ell_1$-minimization \cite{wright2010dense}. The least absolute deviation (LAD) type of algorithms are known to have robust properties in linear regression models (e.g. \cite{rao1995linear}) and are important when heavy-tailed
errors present \cite{wang2013}.  Also, it has been shown to perform particularly well in the presence of \emph{sparse} noise \cite{wright2010dense,foygel,Allerton}.
 \end{itemize}
 Of course, the above list of examples is not exhaustive. Different applications might require different choice of the loss function. For instance, in a scenario where noise is known to be bounded it might be preferable to choose the $\ell_\infty$-norm as the loss function.

\subsection{Tight Performance Analysis}\label{sec:app_tight}

In the setup introduced in Section \ref{sec:app_intro}, the fundamental question of interest is characterizing the 
estimation performance of \eqref{eq:genLASSO}. One possible, and widely used, measure of performance is the \emph{normalized squared error}\footnote{similarly defined measures of performance are considered in the literature under the term of \emph{noise sensitivity}, e.g. \cite{wu2012optimal,donoho2011noise} } $\|\hat\x-\x_0\|_2^2/\|\z\|_2^2$, which quantifies robustness of the estimator. Understanding the behavior of this quantity in terms of the choice of the measurement matrix $\A$, the number of measurements $m$, the convex regularizer $\rr$, the value of the regularizer parameter $\la$ and the unknown signal $\x_0$ itself, is both of theoretical and practical interest. As an example, knowledge of the dependence on $\la$ can provide valuable insights for the challenging task of optimally tuning \eqref{eq:genLASSO}. 
In principle, one is interested in the characterization of the normalized squared error for all possible values of the number of observations $m$. However, of particular interest to us is the \emph{underdetermined} regime, in which the number of measurements $m$ is less than the dimension $n$ of the ambient space of $\x_0$. 

Estimation procedures that can be cast in the generic form of \eqref{eq:genLASSO} have been used in practice since at least twenty years. Inevitably, their theoretical analysis has attracted enormous attention during this period of time. The advances in the study of \emph{noiseless} underdetermined problems, that started with the seminal works  \cite{candes2006robust,donoho2006compressed} under the name of ``compressive sampling", has resulted in a significant progress on our understanding of the performance of \eqref{eq:genLASSO} in the presence of noise (e.g. \cite{candes2007dantzig,bickel,Belloni,negahban2012unified,Raskutti}). Although remarkable, those results characterize the normalized squared error only up to unknown numerical multiplicative constants (order-wise analysis). Only recently, Donoho, Maleki, and Montanari derived 
the first precise formulae predicting the limiting behavior of the $\ell_2^2$-LASSO reconstruction error \cite{donoho2011noise}; a proof of the formulae appeared later by Bayati and Montanari \cite{montanariLasso}. The authors of these references consider the $\ell_2^2$-LASSO with $\ell_1$-regularization, i.i.d Gaussian sensing matrix $\A$ and use the Approximate Message Passing (AMP) framework for the analysis (also see subsequent related works \cite{arianComplex,arianDenoising}).
  Stojnic \cite{StoLASSO} derived precise such results for a constrained version of the LASSO, but most significantly, was the first to introduce the idea of analyzing the normalized error using GMT, upon which Theorem \ref{thm:main} builds.   Theorem \ref{thm:main} can be used as a framework to extend the analysis of \cite{StoLASSO} to a unifying treatment of \eqref{eq:genLASSO}, which results in \emph{asymptotically tight} expressions for the normalized square error. 
Essentially, via Theorem \ref{thm:main} the analysis of \eqref{eq:genLASSO} can be carried out if  one analyzes a simpler optimization program, instead. We illustrate the basic idea and derive the generic format of this simpler optimization in Section \ref{sec:app_gen}. Specializing this to specific choices of the loss function and the regularizer results in concrete characterizations of the normalized error for different instances of \eqref{eq:genLASSO} (LASSO, generalized LASSO, LAD, etc.). The details of this latter part of the analysis is out of the scope of this work. We refer the interested reader to our accompanying papers \cite{OTH,TOH14,LS,ICASSP,Allerton,ISIT15_1,ISIT15_2}. We remark that even though the results hold in the asymptotic of the problem dimensions $m$ and $n$ growing to infinity, numerical simulations show that the expressions are already reasonably tight for values of $m$ and $n$ on a few hundreds. 

For our analysis,  we assume onwards that the measurement matrix $\A$ has entries i.i.d. standard normal. Treating the measurement matrix as generated from a random ensemble is a common practice in the literature of compressive sensing (please refer to the tutorials\cite{vershynin2014estimation,candes2014mathematics}). Specifically the Gaussian ensemble is commonly used as a generic assumption since: i) it is very well understood and enjoys remarkable properties which greatly facilitate the analysis. In our case, such a property is the Gaussian min-max Theorem \ref{lem:Gor}, ii) many of the results derived for the Gaussian ensemble are known to enjoy a \emph{universality} property, i.e. to hold true for fairly broad family of probability ensembles. For instance, empirical simulations performed in \cite{montanariLasso,OTH} suggest that the normalized square error of \eqref{eq:genLASSO} admits the same asymptotic expression even when $\A$ has entries i.i.d. sub-gaussian of zero mean and unit variance. 

Even though not discussed here in any further detail, it is shown in \cite{ISIT15_1} how Theorem \ref{thm:main} can be used to analyze the performance of \eqref{eq:genLASSO} when $\A$ is an isotropically random orthogonal matrix, i.e. is uniformly distributed on the Stiefel manifold, the set of all orthogonal matrices $\mathbf{Q}\in\R^{m\times n}$ such that  $\mathbf{Q}\mathbf{Q}^T = \mathbf{I}_m$.

\subsection{Using the framework of Theorem \ref{thm:main}}\label{sec:app_gen}


\subsubsection{Conjugate pairs}

For the analysis, we  require the notion of the Fenchel conjugate; we briefly recall here its definition and some of its relevant properties. The conjugate of $\phi:\R^k\rightarrow\R$
is the function $\phi^*:\R^k\rightarrow(-\infty,+\infty]$\footnote{ Following the common practice (e.g. as in \cite[Ch.~12]{Roc70} and \cite[Ch.~7]{Bertsekas}) we define $\phi^*(\cdot)$ as an extended real-valued function that takes the value $+\infty$ whenever $\ub\notin\operatorname{dom}\phi$.} defined as $\phi^*(\ub) := \sup_{\vb} \vb^T\ub - \phi(\ub)$. The conjugate function is  convex (as the pointwise supremum of affine functions) and lower semi-continuous. Also, if $\phi(\cdot)$ is convex and continuous (lower semi-continuity suffices) $\phi(\vb) = \sup_{\ub\in\R^k} \{\ub^T\vb - \phi^*(\ub)\}$ for all $\vb\in\dom f$ \cite[Thm. 12.2]{Roc70}. Some standard examples of conjugate pairs of continuous convex functions, also relevant to our analysis, are the following:

$\phi(\vb) = (1/2)\|\vb\|^2$ $\leftrightarrow$ $\phi^*(\ub) = (1/2)\|\ub\|^2$,
~~~~ and, ~~~~ $\phi(\vb) = \|\vb\|$ $\leftrightarrow$ $\phi^*(\ub) = \begin{cases} 0 & \|\ub\|_*\leq 1, \\ +\infty &\text{else}.\end{cases}$

Here, $\|\ub\|_*=\sup_{\|\vb\|\leq 1}\vb^T\ub$ denotes the dual-norm of $\|\cdot\|$. For instance, $\|\cdot\|_\infty$ is the dual-norm of $\|\cdot\|_1$, while $\|\cdot\|_2$ is self-dual.


\subsubsection{Assumptions}

Next, let us formalize our assumptions on the nature of  the loss function $\loss(\cdot):\R^m\rightarrow[0,\infty)$ and the regularizer function $f(\cdot):\R^n\rightarrow[0,\infty)$.
Both are assumed \emph{convex}
, continuous on $\R^m$ and $\R^n$ respectively,
and, coercive\footnote{Adapting the terminology from \cite{Bertsekas}, a function $\phi(\cdot):\R^n\rightarrow\R$ is coercive if for every sequence $\{\x_k\}\subset\R^n$ such that $\|\x_k\|\rightarrow\infty$, we have $\lim_{k\rightarrow\infty}\phi(\x_k)=\infty$}.
  The latter is only a mild technical assumption that conforms with the purpose of the two functions; in particular, all five instances of \eqref{eq:genLASSO} listed in Section \ref{sec:app_intro} satisfy this. In addition, we assume that the conjugate of the loss function $\loss^*(\cdot)$ is continuous in its effective domain $\operatorname{dom}\loss^*:=\{\ub | \loss^*(\ub)<\infty\}$. 
  We have not made any particular effort to relax this technical assumption, partly because it appears to be mild for our interests. In particular, it is satisfied for all  instances of \eqref{eq:genLASSO} considered in  Section \ref{sec:app_intro}.

%

 
\subsubsection{Generic analysis of \eqref{eq:genLASSO}} In order to apply Theorem \ref{thm:main} we need to bring the optimization problem in \eqref{eq:genLASSO} in the format of the (PO) problem in \eqref{eq:1}. But first,  it is convenient to rewrite it after changing the decision variable to the quantity of interest $\w:=\x-\x_0$:
\begin{align}\label{eq:gL1}
\hat\w = \arg\min_{\w} \loss(\A\w-\z)+\la\rr(\x_0+\w),
\end{align}
where we also used $\y=\A\x_0 + \z$. Of course, one cannot actually solve \eqref{eq:gL1} since it involves $\z$ and $\x_0$, which are both unknown; thus, 
it  only serves the purposes of the analysis.
To bring it in the desired format, we express $\loss(\cdot)$ in terms of its conjugate function $\loss^*:\R^m\rightarrow(-\infty,+\infty]$, to obtain:
\begin{align}\label{eq:gL2}
\hat\w = \arg\min_{\w}\sup_{\ub} ~\ub^T\A\w-\ub^T\z - \loss^*(\ub)+ \la\rr(\x_0+\w).
\end{align}
Identify $\f(\w,\ub) := -\ub^T\z-\loss^*(\ub) + \la\rr(\x_0+\w)$ and recall our assumption that $\A$ has entries i.i.d. $\Nn(0,1)$  to see that \eqref{eq:gL2}  is in the desired format of \eqref{eq:1}. This brings us just one step before being able to apply Theorem \ref{thm:main}. The remaining technicality is that  Theorem \ref{thm:main} requires (also see Definition \ref{def:ad}) both the constraint sets of the min-max optimization to be bounded and the function $\f(\w,\ub)$ to be continuous on them. 
Using the (mild) technical assumptions on $\loss(\cdot)$ and $f(\cdot)$, we argue in Appendix \ref{sec:tech} that the set of minima in \eqref{eq:gL1} is compact and for any $\w$ in a compact set, the supremum over $\ub$ in \eqref{eq:gL2} is attained. Thus,
 we can properly choose (sufficiently large)  $K_\w>0$ and $K_\ub>0$ such that 
\begin{align}\label{eq:gL3}
\hat\w = \arg\min_{\|\w\|_2\leq K_\w}\max_{\|\ub\|_2\leq K_\ub} ~\ub^T\A\w-\ub^T\z - \loss^*(\ub)+ \la\rr(\x_0+\w),
\end{align}
and $\loss^*(\ub)$ is continuous for all $\|\ub\|_2\leq K_\ub$.
With these, we may express\footnote{Strictly speaking, to be in accordance with the setup of Theorem \ref{thm:main}, we need to consider a sequence  $\{\A^\mn, \z^\mn, \x_0^\mn, f^\mn(\cdot) \}$, such that $\A^\mn\in\R^{m\times n}$ with entries i.i.d. $\Nn(0,1)$, $\z^\mn\in\R^m$ with entries i.i.d. $\Nn(0,\sigma^2)$, $\x_0^\mn\in\R^n$ and $f^\mn:\R^n\rightarrow\R$ a convex function. With these we can properly define an admissible sequence as in Theorem \ref{thm:main}. We avoid explicitly  introducing this notation in the main text in order to keep the presentation simple.
} the corresponding (AO) problem as:
\begin{align}
\phi(\g,\h) &= \min_{\|\w\|_2\leq K_\w}\max_{\|\ub\|_2\leq K_\ub} \|\w\|_2\g^T\ub + \|\ub\|_2\h^T\w-\z^T\ub-\loss^*(\ub)+\la\rr(\x_0+\w)\nn \\
&=\min_{\|\w\|_2\leq K_\w}\max_{\|\ub\|_2\leq K_\ub} (\|\w\|_2\g-\z)^T\ub - \loss^*(\ub) + \|\ub\|_2\h^T\w+\la\rr(\x_0+\w).\label{eq:gLG}
\end{align}
A few comments are in place. First, observe that the assumptions of statement (ii) of Theorem \ref{thm:main} are satisfied since the constraint sets are both convex and compact, and, the function $\f(\w,\ub) = -\ub^T\z-\loss^*(\ub)+\la\rr(\x_0+\w)$ is convex in $\w$ and concave 
(recall that $\loss^*(\cdot)$ is convex)
 in $\ub$. Thus, $\Phi(\A)$, and, equivalently the optimal cost of \eqref{eq:genLASSO} concentrate around the same quantity to which $\phi(\g,\h)$ does. The next step involves the analysis of the (asymptotic) behavior of $\phi(\g,\h)$. Also, in order to identify $\|\hat\w\|_2$\footnote{We focus on the analysis of the $\ell_2$-norm of the error vector $\w=\x-\x_0$, but Corollary \ref{cor:II} can in principle be applied to analyze other norms, too.}, we still need to verify the conditions of statement (iii) of Theorem \ref{thm:main}, or Corollary \ref{cor:II}. We comment on these in the rest of the section. 
 

\subsubsection{Discussion}

Thus far, we have only introduced a model for the measurement matrix $\A$ and have been silent regarding the noise vector $\z$ and the unknown signal $\x_0$. Having assumed that  $\A$ is random gaussian allowed us to derive a corresponding (AO) problem for \eqref{eq:genLASSO}. For the next step, i.e. the analysis of \eqref{eq:gLG}, we need to model $\z$ and $\x_0$, as well. These depend on the specific instance of \eqref{eq:genLASSO}. For instance, for the LASSO it is typically assumed that $\z$ is gaussian, while a sparse noise model is more reasonable for the  LAD.
  Also, an $\ell_1$-regularizer is typically associated with a sparse $\x_0$, while nuclear-norm regularization corresponds to a low-rank $\x_0$. Thus, the analysis of \eqref{eq:gLG} is problem specific. Note, however, that the probabilistic relation established by Theorem \ref{thm:main} between \eqref{eq:gL3} and \eqref{eq:gLG} holds for all $\z$ and all $\x_0$. Thus, provided that $\A$ is statistically independent from them, Theorem \ref{thm:main} continues to hold even after interpreting the probabilities to be over the joint distribution of $\A$, $\z$ and $\x_0$.

The purpose of this section has been to set up a generic framework and introduce the machinery for the analysis of algorithms that can be cast in the format of \eqref{eq:genLASSO}. 
Of course, a final answer to the problem of interest (here, characterizing the squared-error) is only obtained after the analysis of the (AO) problem as prescribed by the third statement of Theorem \ref{thm:main}. This part, is typically involved on its own and is out of the scope of the paper. Nevertheless, in the next few lines we describe the specific format that  the (AO) problem takes for the five instances of \eqref{eq:genLASSO} that were described in Section \ref{sec:app_intro}, and provide references to works that deal with their specific analysis.

\subsubsection{Examples}
\begin{itemize}
\vp
 \item {LASSO:} Here, $\loss^*(\ub) = (1/2)\|\ub\|^2$, $f(\cdot) = \|\cdot\|_1$, $\z$ has entries i.i.d. $\Nn(0,\sigma^2)$. Note that $\|\w\|_2\g-\z$ in \eqref{eq:gLG} is statistically identical to a random vector with entries i.i.d $\Nn(0,\|\w\|_2^2+\sigma^2)$ and with some abuse of notation \eqref{eq:gLG} takes the form:
\begin{align}
\phi_{\text{LASSO}}(\g,\h) &=\min_{\|\w\|_2\leq K_\w}\max_{\|\ub\|_2\leq K_\ub} \sqrt{\|\w\|_2^2+\sigma^2}~\g^T\ub - \frac{1}{2}\|\ub\|_2^2 + \|\ub\|_2\h^T\w+\la\|\x_0+\w\|_1\nn
\end{align}
where $\g\in\R^m$, $\h\in\R^n$ have entries i.i.d. $\Nn(0,1)$. 
\vp
 \item {$\ell_2$-LASSO:}  We have $\dom\loss^* = \{\ub~|~\|\ub\|_2\leq 1\}$ and $\loss^*(\ub)= 0$ for all $\ub\in\dom\loss^*$. Also, the entries of $\z$ are i.i.d. $\Nn(0,\sigma^2)$. Hence, \eqref{eq:gLG} corresponding to the $\ell_2$-LASSO with $\ell_1$-regularization becomes:
 \begin{align}
\phi_{\ell_2\text{-LASSO}}(\g,\h) &=\min_{\|\w\|_2\leq K_\w}\max_{\|\ub\|_2\leq 1} \sqrt{\|\w\|_2^2+\sigma^2}~\g^T\ub \|\ub\|_2\h^T\w+\la\|\x_0+\w\|_1\label{eq:fin}.
\end{align}
Please refer to \cite{ICASSP} for a detailed treatment of \eqref{eq:fin}, which results in a tight expression for the normalized squared error (NSE) of the $\ell_2$-LASSO with $\ell_1$-regularization. Therein, it is assumed that $\x_0$ is $k$-sparse with non-zero entries i.i.d. $\Nn(0,\sigma_x^2)$. The result is derived for the \emph{linear} regime in which $m/n\rightarrow \delta\in(0,\infty)$ and $k/n\rightarrow \rho\in(0,1)$ as $n\rightarrow\infty$; this regime is assumed throughout all references discussed in this section. The derived formulae explicitly characterizes the NSE for \emph{all} values of $\sigma^2$ as  a function of $m$, $n$, $k$, $\sigma_x^2$ and $\la$. 

\vp
 \item {Generalized-LASSO:} Here, $f(\cdot)$ is arbitrary. As above, $\z$ has entries i.i.d. $\Nn(0,\sigma^2)$ and the loss function can be either $(1/2)\|\cdot\|^2$ or $\|\cdot\|_2$. We perform a \emph{high-SNR} analysis in the limit of $\sigma^2\rightarrow 0$. In this regime, we can approximate $f(\cdot)$ on the first-order (\cite[Thm.~23.4]{Roc70}) as $f(\x_0+\w)\approx f(\x_0) + \max_{\s\in\paf} \s^T\w$ (please see \cite[Sec.~9.1]{OTH} for details on the validity of the approximation). With these, the (AO) problem for the generalized LASSO and generalized $\ell_2$-LASSO writes:
\begin{align}
\phi_{\text{gen-LASSO}}(\g,\h) &=\min_{\|\w\|_2\leq K_\w}\max_{\substack{\|\ub\|_2\leq K_\ub \\ \s\in\paf }} \sqrt{\|\w\|_2^2+\sigma^2}~\g^T\ub - \frac{1}{2}\|\ub\|_2^2 + (\|\ub\|_2\h + \la\s)^T\w\label{eq:gen_LASSO},
\end{align}
and 
\begin{align}
\phi_{\text{gen-$\ell_2$-LASSO}}(\g,\h) &=\min_{\|\w\|_2\leq K_\w}\max_{\substack{\|\ub\|_2\leq 1 \\ \s\in\paf }} \sqrt{\|\w\|_2^2+\sigma^2}~\g^T\ub + (\|\ub\|_2\h + \la\s)^T\w\label{eq:gen_ell2_LASSO},
\end{align}
respectively.
In the high-SNR regime, the analysis only depends on $f(\cdot)$ and $\x_0$ through a ``first-order surrogate", namely the subdifferential $\paf$. For example, in sparse recovery with $\ell_1$-regularization, the high-SNR normalized square error depends only on the sparsity of the unknown signal $\x_0$. Similarly, in the case of recovery of a low-rank matrix $\x_0$ using nuclear-norm regularization, the high-SNR NSE depends only on the rank of $\x_0$. On the other hand, the NSE of the LASSO in the finite-SNR regime depends on the specific statistics of $\x_0$. For example, the analysis of \eqref{eq:fin} in \cite{ICASSP} is performed under the assumption that $\x_0$ is $k$-sparse with non-zero entries being i.i.d. $\Nn(0,\sigma_x^2)$.

The analysis of \eqref{eq:gen_ell2_LASSO} is performed in \cite{OTH} and that of \eqref{eq:gen_LASSO} in \cite{ISIT15_2}. These yield sharp formulae for the high-SNR NSE of the generalized LASSO and $\ell_2$-LASSO. In fact, extensive numerical simulations and partial theoretical results (\cite[Sec.~10]{OTH},\cite{ICASSP}) that the high-SNR NSE corresponds to the worst-case NSE, i.e. $\lim_{\sigma^2\rightarrow 0}\NSE(\sigma) = \sup_{\sigma^2>0} \NSE(\sigma)$. In this sense, the results of \cite{OTH,ISIT15_2} correspond to tight upper bounds on the NSE for all values of $\sigma^2$ for arbitrary regularizer functions. As a side result, the derived sharp formulae provide guidelines for the optimal tuning of the regularizer parameter $\la$. Also, they yield explicit characterization of the mapping between the regularizer parameters of the LASSO and the $\ell_2$-LASSO. Finally, please refer to \cite{LS} for an interpretation of those results as a natural extension to the classical analysis of the ordinary least squares.

\vp
 \item {LAD}: We have $\dom\loss^* = \{\ub~|~\|\ub\|_\infty\leq 1\}$ and $\loss^*(\ub)= 0$ for all $\ub\in\dom\loss^*$. Here, it is natural to assume that  the noise vector $\z$ is sparse. For instance, assuming that $\z$ is $s$-sparse  (w.l.o.g., $\z_{s+1}=\ldots=\z_m = 0$) with its non-zero entries i.i.d. $\Nn(0,\sigma^2)$, \eqref{eq:gLG} (after the usual first-order approximation of $f(\cdot)$) becomes:
\begin{align}
\phi_{\text{LAD}}(\g,\h) &=\min_{\|\w\|_2\leq K_\w}\max_{\substack{\|\ub\|_\infty\leq 1 \\ \s\in\paf }} \sqrt{\|\w\|_2^2+\sigma^2}\sum_{i=1}^{s}~\g_i\ub_i + \|\w\|_2\sum_{i=s+1}^m\g_i\ub_i + (\|\ub\|_2\h + \la\s)^T\w\label{eq:LAD}
\end{align} 
The analysis of \eqref{eq:LAD} is performed in \cite{Allerton}. With this and an application of Theorem \ref{thm:main}, we have derived tight asymptotic expressions for the high-SNR NSE of the regularized-LAD estimator. Among others, this allows for an exact performance comparison between the LASSO and the LAD (cf. \cite[Sec.II-F]{Allerton}).
 \end{itemize}

It is worth commenting on the sharpness of expressions for the NSE of \eqref{eq:genLASSO}  that are derived through application of Theorem \ref{thm:main}, for finite problem dimensions. Statement (iii) of Theorem \ref{thm:main} requires the problem dimensions to grow to infinity. However, numerical simulations for a wide class of examples (\cite{Sto,OTH,Allerton,ICASSP,ISIT15_2}) suggest that the theoretical predictions become fairly tight for relatively small problem dimensions (ranging over a few hundreds).


%


%
%

\subsection{An example: Working out the details}\label{sec:details}

Assume the same setup as in Section \ref{sec:app_intro} and consider the following instance of \eqref{eq:genLASSO}:
\begin{align}\label{eq:conL}
\hat\x := \arg\min_{\x} \|\y-\A\x\|_2 \text{ s.t. } f(\x)\leq f(\x_0).
\end{align}
Strictly speaking, \eqref{eq:conL} is not an instance of \eqref{eq:genLASSO} since it assumes extra prior information, i.e. knowledge of $f(\x_0)$. In fact, it corresponds to a ``constrained" version of the generalized LASSO. Note, that Lagrange duality ensures that there exists value of the regularizer parameter $\la$ in the (unconstrained) generalized LASSO such that the two problems are equivalent. 
As typical in the analysis of the performance of the LASSO, assume that $\z$ has entries i.i.d $\Nn(0,\sigma^2)$ and is independently generated from $\A$. The goal is to characterize the normalized error $\operatorname{NSE}= \|\hat\x-\x_0\|_2^2/\sigma^2$ of \eqref{eq:conL}. In particular, to facilitate the analysis  we focus in the high-SNR regime in which $\sigma^2\rightarrow 0$ (see \cite{ICASSP} for the general case). It can be shown, that this corresponds to the worst-case NSE, i.e. $\sup_{\sigma>0} \operatorname{NSE} = \lim_{\sigma\rightarrow 0} \operatorname{NSE}$ \cite[Sec. 10]{OTH}. It is also shown in \cite[Sec. 7]{OTH}, that in the high-SNR regime, the NSE is same as the NSE of the first-order approximation of \eqref{eq:conL}, obtained after relaxing the constraint set $\Dc_f(\x_0) = \{\vb \ | \ f(\x_0+\vb)\leq f(\x_0)\}$ to (essentially) its conic hull $\Tc_f(\x_0) = \operatorname{Cl}(\operatorname{cone}(  \Dc_f(\x_0) )$,
where $\operatorname{Cl}(\cdot)$ denotes the set closure operator and $\operatorname{cone}(\cdot)$ returns the conic hull of a set. The rationale behind this approximation being tight is that as $\sigma\rightarrow 0$, the squared error also approaches zero. Thus, onwards, we focus on the following:
\begin{align}\label{eq:conL1}
\hat\x := \arg\min_{\x} \|\y-\A\x\|_2 \quad\text{ s.t. }~~ \x-\x_0\in\Tc_f(\x_0).
\end{align}
Adapting the results of Section \ref{sec:app_gen} to \eqref{eq:conL1}, the corresponding (AO) problem becomes:
\begin{align}
\phi(\g,\h,\z)
&=\min_{\substack{\w\in\Tc_f(\x_0)\\\|\w\|_2\leq K}}\max_{\|\ub\|_2\leq 1} (\|\w\|_2\g-\z)^T\ub + \|\ub\|_2\h^T\w.\label{eq:opaopa0}
\end{align}
Recall that $\z$ is assumed to have entries i.i.d $\Nn(0,\sigma^2)$ and be independent of $\A$. Thus, $\|\w\|_2\g-\z$ is statistically identical to a random vector with entries i.i.d $\Nn(0,\|\w\|_2^2+\sigma^2)$ and with some abuse of notation we can rewrite \eqref{eq:opaopa} as
\begin{align}
\phi(\g,\h)
&=\min_{\substack{\w\in\Tc_f(\x_0)\\\|\w\|_2\leq K}}\max_{\|\ub\|_2\leq 1} \sqrt{\|\w\|_2^2+\sigma^2}~\g^T\ub + \|\ub\|_2\h^T\w,\label{eq:opaopa}
\end{align}
where, again, $\g\in\R^m$ has entries i.i.d. $\Nn(0,1)$.

\underline{Remark}: In \eqref{eq:opaopa}, we have constrained $\w$ to a bounded set. In \eqref{eq:conL} the set of minima is trivially a bounded set, since the feasible set is itself bounded. Notice, however,  that this is not the case with the ``relaxed" problem in \eqref{eq:conL1}. Thus, the argument presented in Appendix \ref{sec:tech} to which we appealed in Section \ref{sec:app_gen} does not directly apply here. Instead, we rely on the following argument. Assume that $\Nc(\A)\cap\Tc_f(\x_0) = \{0\}$. When this holds, then the objective function and constraint set in \eqref{eq:conL1} have no common nonzero directions of recession, hence the set of minima is a compact set \cite[Prop.~2.3.2]{Bertsekas}. Thus, under this assumption introducing the boundedness constraint in \eqref{eq:opaopa} is validated. It can be shown that the desired condition is true with overwhelming probability provided that $m>\omega^2(\Tc_f(\x_0)\cap\Sc^{n-1})$, where $\omega$ is a geometric measure of the size of $\Tc_f(\x_0)$, defined later. Hence, in the regime  $m>\omega^2(\Tc_f(\x_0)\cap\Sc^{n-1})$ the solution set of \eqref{eq:conL1} is bounded w.h.p. and the analysis of \eqref{eq:opaopa} is validated. 

\vp
In what follows, we apply the recipe prescribed by Corollary \ref{cor:II} to \eqref{eq:opaopa}. First, we need to normalize \eqref{eq:opaopa} by dividing with $\sqrt{m}$. This corresponds to the same operation applied to the original problem \eqref{eq:conL1}, and, has no effect on the value of the NSE. Next, we work sequentially on meeting the four assumptions of the corollary.

\vp
\emph{(1)}
By its definition in \eqref{eq:xfun} and by \eqref{eq:opaopa} (recall the normalization with $\sqrt{m}$ and that we drop the superscript $(n)$ for simplicity)
$$
\upsilon(\w;\g,\h) = \frac{1}{\sqrt{m}} \max_{\|\ub\|\leq 1}\sqrt{\|\w\|_2^2+\sigma^2}~\g^T\ub + \|\ub\|_2\h^T\w.
$$
Easily, this can be simplified to a scalar maximization:
$$
\upsilon(\w;\g,\h) =  \frac{1}{\sqrt{m}} \max_{0\leq\beta\leq 1}\sqrt{\|\w\|_2^2+\sigma^2}\|\g\|\beta + \beta\h^T\w.
$$
Observe that the function $\nu(\w,\beta;\g,\h)=(1/\sqrt{m})(\sqrt{\|\w\|_2^2+\sigma^2}\|\g\|\beta + \beta\h^T\w)$ above is convex in $\w$ and linear (thus, concave) in $\beta$. Furthermore, the constraint imposed on $\beta$ is convex and compact.

\vp
\emph{(2)}
We assume that $K$ in \eqref{eq:opaopa} is constant, i.e. it does not scale with $n$. The function $\phi_{\|\cdot\|_2}(\cdot,\beta;\g,\h)$ becomes
\begin{align}
\phi_{\|\cdot\|_2}(\cdot,\beta;\g,\h) &= \frac{1}{\sqrt{m}}\min_{\substack{\w\in\Tc_f(\x_0) \\ \|\w\|_2=\alpha }}\sqrt{\|\w\|_2^2+\sigma^2}\|\g\|\beta + \beta\h^T\w \nn\\
&=\frac{1}{\sqrt{m}}\sqrt{\alpha^2+\sigma^2}\|\g\|\beta -\alpha\beta \max_{\substack{\w\in\Tc_f(\x_0) \\ \|\w\|_2=1 }}(-\h)^T\w\label{eq:check}
\end{align}
We need to check convexity of \eqref{eq:check} with respect to $\alpha$. This follows from the second derivative test since
\begin{align}
 \frac{\partial^2 \phi_{\|\cdot\|_2}(\alpha,\beta;\g,\h)}{\partial \alpha^2} = \frac{\sigma^2\|\g\|_2}{\sqrt{m}(\alpha^2+\sigma^2)^{3/2}}\label{eq:secder}
\end{align}
is nonnegative for all $\alpha\in[0,K]$.

\vp
\emph{(3)}
We will compute deterministic function $d:[0,K]\rightarrow\R$ such that 
\beq
\hat\phi_{\|\cdot\|_2}\phi(\alpha;\g,\h):=\max_{0\leq\beta\leq 1}\phi_{\|\cdot\|_2}\phi(\alpha,\beta;\g,\h)\label{eq:2simp}
\eeq
converges to it point-wise for all $\alpha\in[0,K]$.
Let $(\chi)_+:=\max\{\chi,0\}$ for any $\chi\in\R$ and define
$$
\D(\h) := \max_{\w\in\Tc_f(\x_0), \|\w\|_2=1}\h^T\w.
$$
It is easy in view of \eqref{eq:check}  to simplify \eqref{eq:2simp} as
\begin{align}
\hat\phi_{\|\cdot\|_2}(\alpha;\g,\h) = (\phi(\alpha;\g,\h))_+ := \left( \sqrt{\alpha^2+\sigma^2}\frac{\|\g\|_2}{\sqrt{m}} - \alpha \frac{\D(-\h)}{\sqrt{m}}\right)_+ \label{eq:hat}
\end{align}
Consider $d:[0,K]\rightarrow\R$ defined as follows
$$
d(\alpha) := \sqrt{\alpha^2+\sigma^2}\frac{\gamma_m}{\sqrt{m}} - \alpha \frac{\gw}{\sqrt{m}},
$$
where 
\beq
\gamma_m := \E \|\g\|_2 \quad \text{ and } \quad \gw:=\omega(\Tc_f(\x_0)\cap\Sc^{n-1}):= \E \D(\h). \nn
\eeq
It is well known that $\frac{m}{\sqrt{m+1}}\leq \gamma_m \leq \sqrt{m} $ and $\omega(\Tc_f(\x_0)\cap\Sc^{n-1})$ is known as the ``gaussian width"
\footnote{
The gaussian width $\gw$ appears as a fundamental quantity in the study of noiseless compressed sensing, where one wishes to recover an unknown structured signal $\x_0\in\R^n$ from $m<n$ linear equations via 
$
\min f(\x) \text{ s.t. } \A\x = \A\x_0.
$
Earlier works \cite{stojnic2009various,Cha} had proved that $m>\gw^2$ number of measurements suffice for this convex algorithm to uniquely recover $\x_0$. More recently, it was shown independently in \cite{TroppEdge,stojnic2013upper} that $\gw^2$ number of measurements are also necessary for unique recovery. The arguments in \cite{stojnic2013upper} rely on GMT, while \cite{TroppEdge} uses tools from conic integral geometry; see \cite{amelunxen2014gordon} for a connection between those two. 
 It is important to note that the gaussian width $\gw$ admits accurate estimates for a number of important regularizers $f(\cdot)$. For example, for $f(\cdot)=\|\cdot\|_1$ and $\x_0$ k-sparse, it is shown in \cite{Cha,TroppEdge} that $\gw^2\lesssim 2k\log(2n/k)$. See \cite{Cha,TroppEdge,foygel} for more examples.
}
 of $\Tc_f(\x_0)\cap\Sc^{n-1}$. The motivation behind choosing $d(\cdot)$ as a candidate to establish convergence is clear: both functions $\|\cdot\|_2$ and $d(\cdot)$ are 1-Lipschitz (e.g. \cite{gorLem}), thus, from the Gaussian concentration of measure (Proposition \ref{prop:lip}) they concentrate around their means. In order to be able to formally state our convergence results, we need to first specify the scaling of the problem dimensions as they grow to infinity. We assume the ofter called \emph{linear regime} in which $n,m,\gw$ grow to infinity in proportional rates:
 $$
 m/n\rightarrow \zeta\in(0,\infty) \qquad \text{ and } \qquad \text{ $(1-\eps)\gamma_m>\gw>\eps\gamma_m$, for some $0<\eps<1$.}
 $$
 With these, for all $\alpha\in[0,K]$:
\beq
 \phi(\alpha;\g,\h)\rP d(\alpha).\label{eq:c1}
\eeq
 Furthermore, observe that $\gamma_m>\gw$ implies strict non-negativity of $d(\cdot)$. Using this and \eqref{eq:c1}, the desired conclusion follows from \eqref{eq:hat}:
 $$
 \hat\phi_{\|\cdot\|_2}(\alpha;\g,\h) \rP d(\alpha),\quad \forall\alpha\in[0,K].
 $$
\vp
\emph{(4)} Strong convexity of $d(\cdot)$ follows immediately by the second derivative test:
 $$
 \frac{\partial^2 d(\alpha)}{\partial \alpha^2} = \frac{\sigma^2\gamma_m}{\sqrt{m}(\alpha^2+\sigma^2)^{3/2}} \geq \frac{\sigma^2\gamma_m}{\sqrt{m}(K^2+\sigma^2)^{3/2}}>0.
 $$
Furthermore, setting the derivative of $d(\cdot)$ to zero, yields\footnote{At this point recall that the value of the constant $K$ that constraints the optimization variable $\alpha$, can be set arbitrarily. In particular, set $K>\sigma\gw/\sqrt{\gamma_m^2-\gw^2}$.}
$$
\alpha_* = \sigma\frac{\gw}{\sqrt{\gamma_m^2-\gw^2}} \quad \text{ and } \quad d_*:=d(\alpha_*)=\sigma\frac{\sqrt{\gamma_m^2-\gw^2}}{\sqrt{m}}.
$$

\vp
We may now apply the conclusion of Corollary \ref{cor:II}: when $m/n\rightarrow\zeta\in(0,\infty)$ and  $(1-\eps)\gamma_m>\gw>\eps\gamma_m$ for some $\eps\in(0,1)$,
we have
\begin{subequations}\label{eq:thatsIT}
\beq
\lim_{\sigma\rightarrow 0}\frac{\|\y-\A\hat\x\|_2}{\sqrt{m}\sigma} \rP \frac{\sqrt{{\gamma_m^2-\gw^2}}}{\sqrt{m}},
\eeq
\beq
\lim_{\sigma\rightarrow 0}\frac{\|\hat\x-\x_0\|_2^2}{\sigma^2} \rP \frac{\gw^2}{{\gamma_m^2-\gw^2}}.
\eeq
\end{subequations}

Note that \eqref{eq:thatsIT} holds for \emph{arbitrary} structured signals and associated convex  regularizers $f$. In this generality the result was first proved in \cite[Thm. 3.1]{OTH}. The work in \cite{OTH} and in part the current work are highly motivated and build upon the ideas introduced by Stojnic in \cite{StoLASSO}. Stojnic was the first to apply GMT and combine it with a duality argument, and, was able to prove \eqref{eq:thatsIT} for sparse $\x_0$ and $\ell_1$-regularization \cite[Thm. 1]{StoLASSO}. When compared to \cite{StoLASSO} and \cite{OTH}, the derivation of the result as presented in this section is significantly simplified, shortened and insightful. This is due to the machinery offered by Theorem \ref{thm:main}. More significantly, this machinery allowed for the unified treatment of \eqref{eq:genLASSO}.
Continuing with our discussion on the relevant literature, and to the best of our knowledge, the characterization of the NSE in \eqref{eq:thatsIT} for the special case of sparse $\x_0$ and $\ell_1$-regularization first appeared in \cite{donoho2011noise} and was rigorously proved in \cite{bayati2011dynamics} using a framework different than that offered by the GMT.


\section{Related Work and Conclusion}\label{sec:rel}


Starting from the work of Vershynin and Rudelson \cite{rudelson2006sparse}, Gaussian comparison theorems have played instrumental role in developing a 
clear
 understanding of linear inverse problems when the measurement matrix follows the standard Gaussian distribution.
  The idea of combining strong duality with the Gaussian min-max theorem (\GMT)~is originally attributed to Stojnic \cite{stojnic2013meshes}. In a recent line of work he makes repeated use of this powerful idea. In \cite{stojnic2013upper} he applies it to prove that the $\ell_1$-minimization phase transition thresholds of \cite{stojnic2009various,Cha} are tight. A similar observation also appears in \cite{TroppEdge} by Amelunxen et.al.. In these works, the strong duality argument originates from the KKT optimality conditions rather than swapping min-max. Furthermore, Stojnic applies this idea to prove a tight upper bound on the normalized squared error of the LASSO algorithm with $\ell_1$ regularization \cite{StoLASSO}. The result was later generalized and extended in various directions by the current authors (see references in Section \ref{sec:app_tight}). Finally, Stojnic showed how similar ideas can be applied to the  study of the storage capacity of perceptrons \cite{stojnic2013spherical}.

This work is motivated by and builds upon Stojnic's original idea. Our insights and  additional technical effort lead to a succinct statement  of our main result in Theorem \ref{thm:main} and Corollaries \ref{cor:main} and \ref{cor:II}, which all appear to be novel. In Theorem \ref{thm:main} we have quantified explicit (sufficient) conditions that are required for the GMT to be tight. A critical observation amounts to the fact that through a symmetrization trick we can get rid of the term $g\|\x\|_2\|\y\|_2$ in one of the Gaussian processes involved in \GMT. The resulting minimax optimization problem is now convex and the rest follows. 
In Section \ref{sec:app} we showed the power of Theorem \ref{thm:main} by applying it to pinpoint the optimal cost of the LASSO optimization. In particular, we were able to recover a result from \cite{OTH} with substantially less effort and through a more insightful treatment. The direct and simplified nature of Corollary \ref{cor:II}, when compared to the rather complex arguments in \cite{StoLASSO,OTH}, allows for a unifying treatment of \eqref{eq:intro_gen}. 






\bibliography{compbib}

\appendix

\subsection{Gordon's \thename}\label{sec:app_Gor}
Gaussian comparison theorems are powerful tools in probability theory \cite{ledoux}.  A particularly useful such comparison inequality is described by Gordon's comparison theorem. In fact Gordon's theorem, is a generalization of the classical Slepian lemma and Fernique theorem \cite{gorThm}. It was first proved by Y. Gordon in \cite{gorThm}, where it was also shown how it can be used as an alternative to (re)-derive other well-known results in the field. See also \cite{gorGen} for slight generalized versions of the theorem and the classical reference \cite[Chapter~3.3]{ledoux}  for an introduction to gaussian comparison theorems and some applications.

\begin{appthm}[Gordon's Gaussian comparison theorem, \cite{gorThm}]\label{thm:GordonMain}
Let $\left\{X_{ij}\right\}$ and $\left\{Y_{ij}\right\}$, $1\leq i\leq I$, $1\leq j\leq J$, be  centered Gaussian processes such that
\begin{align*}
\begin{cases}
\E X_{ij}^2 = \E Y_{ij}^2, &\text{ for all } i,j,\\
\E X_{ij}X_{ik} \geq \E Y_{ij}Y_{ik}, &\text{ for all } i,j,k,\\
\E X_{ij}X_{\ell k} \leq \E Y_{ij}Y_{\ell k}, &\text{ for all } i\neq\ell \text{ and } j,k.
\end{cases}
\end{align*}
Then, for all  $\lambda_{ij}\in\mathbf{R}$,
$$
\Pro\left( \bigcap_{i=1}^{I} \bigcup_{j=1}^{J} \left[ Y_{ij} \geq \lambda_{ij} \right] \right) \geq \Pro\left( \bigcap_{i=1}^{I} \bigcup_{j=1}^{J} \left[ X_{ij} \geq \lambda_{ij} \right] \right).
$$
\end{appthm}

Gordon's Theorem \ref{thm:GordonMain} establishes a probabilistic comparison between two abstract Gaussian processes $\{X_{ij}\}$ and $\{Y_{ij}\}$ based on conditions on their corresponding covariance structures. Proposition \ref{lem:Gor} is a corollary of Theorem \ref{thm:GordonMain} when applied to specific Gaussian processes.

We begin with using Theorem \ref{thm:GordonMain} to prove an analogue of Proposition \ref{lem:Gor} for discrete sets. The proof is almost identical to the proof of Gordon's original Lemma $3.1$ in \cite{gorLem}. Nevertheless, we include it here for completeness. After the proof of Lemma \ref{lem:GorD}, we use a compactness argument to translate the result to continuous sets and complete the proof of Proposition \ref{lem:Gor}. 

To simplify notation we suppress notation and write $\|\cdot\|$ instead of $\|\cdot\|_2$.

\begin{applem}[Gordon's \thename: Discrete Sets]\label{lem:GorD}
Let $\A\in\mathbb{R}^{{m}\times {n}}$, $g\in\R$, $\g\in\mathbb{R}^{m}$ and $\h\in\mathbb{R}^{n}$ have  entries i.i.d. $\Nn(0,1)$ and be independent of each other. Also, let $\Ic_1\subset\R^{n}$, $\Ic_2\subset\R^{m}$ be finite sets of vectors and $\psi(\cdot,\cdot)$ be a finite function defined on $\Ic_1\times\Ic_2$.
Then, for all $c>0$:
\begin{align*}
\Pro\left(  \min_{\x\in\Ic_1}\max_{\y\in\Ic_2}~\left\{ \y^T\A\x + g\|\x\|\|\y\| + \psi({\x,\y}) ) \right\} \geq c\right) \geq 
\Pro\left(  \min_{\x\in\Ic_1}\max_{\y\in\Ic_2}\left\{ \|\x \| \g^T\y  +  \|\y\| \h^T\x + \psi({\x,\y}))  \right\}\geq c  \right)
\end{align*}
\end{applem}

\begin{proof}
Define two Gaussian processes indexed on the set $\Ic_1\times \Ic_2$:
\begin{align*}
Y_{\x,\y} = \x^T \Gb \y + g\|\y\| \|\x\| \quad \text{ and }\quad
X_{\x,\y} =  \|\x \| \g^T\y  -  \|\y\| \h^T\x.
\end{align*}

First, we show that the processes defined satisfy the conditions of Gordon's Theorem \ref{thm:GordonMain}. Clearly, they are both centered. Furthermore, for all $\x,\x'\in\Ic_1$ and $\y,\y'\in\Ic_2$:
\begin{align*}
\E[X_{\x,\y}^2] =  \|\x \|^2 \|\y\|^2 + \|\y\|^2 \|\x\|^2 = \E[Y_{\x,\y}^2],
\end{align*}
and
\begin{align*}
\E[ X_{\x,\y}X_{\x',\y'} ] - \E[ Y_{\x,\y}Y_{\x',\y'} ]  &= \|\x\|\|\x'\|(\y^T\y') + \|\y\|^2(\x^T\x') -  (\x^T\x') (\y^T\y') - \|\y\|\|\y'\| \|\x\|\|\x'\| \notag \\
&= \left(\underbrace{ \|\x\|\|\x'\| -  (\x^T\x') }_{\geq 0} \right) \left( \underbrace{ (\y^T\y') - \|\y\|\|\y'\| }_{\leq 0} \right),
\end{align*}
which is non positive  and equal to zero when $\x=\x'$. 

Next, for each $(\x,\y)\in \Ic_1\times\Ic_2$, let $\la_{\x,\y} = -\psi(\x,\y) + c$ and apply Theorem \ref{thm:GordonMain}. This completes the proof by observing that
$$
\left[\min_{\x\in\Ic_1}\max_{\y\in\Ic_2} \{Y_{\x,\y} + \psi(\x,\y) \} \geq c\right] = \bigcap_{\x\in\Ic_1}\bigcup_{\y\in\Ic_2}\left[ Y_{\x,\y} \geq \la_{\x,\y} \right],
$$
and similar for the process  $X_{\x,\y}$.
\end{proof}

\begin{proof}(of Proposition \ref{lem:Gor})
Denote $R_1:=\max_{\x\in\Sc_\x}\|\x\|$ and  $R_2:=\max_{\y\in\Sc_\y}\|\y\|$.
Fix any $\eps>0$.
Since $\psi(\cdot,\cdot)$ is continuous and the sets $\Sc_\x$,$\Sc_\y$ are compact, $\psi(\cdot,\cdot)$ is uniformly continuous on $\Sc_\x\times\Sc_\y$. Thus, there exists $\delta:=\delta(\eps)>0$ such that for every $(\x,\y),(\tilde\x,\tilde\y)\in\Sc_\x\times\Sc_\y$ with $
\|\begin{bmatrix}\x & \y\end{bmatrix}-\begin{bmatrix}\tilde\x & \tilde\y\end{bmatrix}\|\leq \delta$, we have that $|\psi(\x,\y)-\psi(\tilde\x,\tilde\y)|\leq \eps$. Let $\Sc_\x^\delta$,$\Sc_\y^\delta$ be $\delta$-nets of the sets $\Sc_\x$ and $\Sc_\y$, respectively. Then, for any $\x\in\Sc_\x$, there exists $\x'\in\Sc_\x^\delta$ such that $\|\x-\x'\|\leq\delta$ and an analogous statement holds for $\Sc_\y$. In what follows, for any vector $\vb$ in a set $\Sc$, we denote $\vb'$ the element in the $\delta$-net of $\Sc$ that is the closest to $\vb$ in the usual $\ell_2$-metric.
To simplify notation, denote
$$\alpha(\x,\y) := \y^T\A\x + g\|\x\|\|\y\| + \psi({\x,\y}) \quad \text{ and } \quad \beta(\x,\y) : =  \|\x \| \g^T\y  +  \|\y\| \h^T\x + \psi({\x,\y}).$$
From Lemma \ref{lem:GorD}, we know that for all $c\in\R$:
\begin{align}\label{eq:weknow}
\Pro\left(  \min_{\x\in\Sc_\x^\delta}\max_{\y\in\Sc_\y^\delta}~\alpha(\x,\y) \geq c\right) \geq 
\Pro\left(  \min_{\x\in\Sc_\x^\delta}\max_{\y\in\Sc_\y^\delta} \beta(\x,\y) \geq c  \right).
\end{align}
In what follows we show that constraining the minimax optimizations over only the $\delta$-nets $\Sc_\x^\delta,\Sc_\y^\delta$ instead of the entire sets $\Sc_\x$,$\Sc_\y$, changes the achieved optimal values by only a small amount. 

First, we calculate an upper bound on 
\begin{align*}
 \min_{\x\in\Sc_\x^\delta}\max_{\y\in\Sc_\y^\delta}\alpha(\x,\y) - \min_{\x\in\Sc_\x}\max_{\y\in\Sc_\y}\alpha(\x,\y) &\leq  \min_{\x\in\Sc_\x^\delta}\max_{\y\in\Sc_\y^\delta}\alpha(\x,\y) - \min_{\x\in\Sc_\x}\max_{\y\in\Sc_\y^\delta}\alpha(\x,\y)
=:\alpha(\x_1,\y_1) - \alpha(\x_2,\y_2)\\
&\leq {\max_{\y\in\Sc_\y^\delta}\alpha(\x_2',\y)} - \alpha(\x_2,\y_2)=:\alpha(\x_2',\y_*) - \alpha(\x_2,\y_2)\\
&\leq \alpha(\x_2',\y_*) - \alpha(\x_2,\y_*) \\
&= \y_*^T\A(\x_2'-\x_2) + g\|\y_*\|(\|\x_2'\|-\|\x_2\|) + (\psi(\x_2',\y_*)-\psi(\x_2,\y_*))\\
&\leq  (\|\A\|_2+|g|)\underbrace{\|\y_*\|}_{\leq R_2}\underbrace{\|\x_2'-\x_2\|}_{\leq\delta} + \underbrace{|\psi(\x_2',\y_*)-\psi(\x_2,\y_*)|}_{\leq\eps}\\
&\leq (\|\A\|_2+|g|)R_2\delta + \eps.
\end{align*}
From this, we have that
\begin{align}\label{eq:rel1}
\Pro\left( \min_{\x\in\Sc_\x}\max_{\y\in\Sc_\y}\alpha(\x,\y) \geq c  \right) \geq \Pro\left( \min_{\x\in\Sc_\x^\delta}\max_{\y\in\Sc_\y^\delta}\alpha(\x,\y) \geq c + (\|\A\|_2+|g|)R_2\delta + \eps \right).
\end{align}
Using standard concentration results on Gaussians, it is shown in Lemma \ref{lem:conc1} that for all $t>0$, $$\Pro(\|\A\|_2+|g|\leq \sqrt{m}+\sqrt{n}+1 + t)\geq 1- 2\exp(-t^2/4).$$
This, when combined with \eqref{eq:rel1} yileds:
\begin{align}\label{eq:vrika1}
\Pro\left( \min_{\x\in\Sc_\x}\max_{\y\in\Sc_\y}\alpha(\x,\y) \geq c  \right) \geq \Pro\left( \min_{\x\in\Sc_\x^\delta}\max_{\y\in\Sc_\y^\delta}\alpha(\x,\y) \geq c +(\sqrt{n}+\sqrt{m}+1+t)R_2\delta + \eps \right) - 2\exp(-t^2/4).
\end{align}

Similarly,
\begin{align*}
 \min_{\x\in\Sc_\x^\delta}\max_{\y\in\Sc_\y^\delta}\beta(\x,\y) - \min_{\x\in\Sc_\x}\max_{\y\in\Sc_\y}\beta(\x,\y) 
 &\geq  \min_{\x\in\Sc_\x^\delta}\max_{\y\in\Sc_\y^\delta}\beta(\x,\y) - \min_{\x\in\Sc_\x^\delta}\max_{\y\in\Sc_\y}\beta(\x,\y)
=:\beta(\x_1,\y_1) - \beta(\x_2,\y_2)\\
&\geq \beta(\x_1,\y_1) - \max_{\y\in\Sc_\y}\beta(\x_1,\y) =:\beta(\x_1,\y_1) - \beta(\x_1,\y_*)\\
&\geq \beta(\x_1,\y_*') - \beta(\x_1,\y_*) \\
&=  \|\x_1\|\g^T(\y_*'-\y_*) + (\|\y_*'\|-\|\y_*\|)\h^T\x_1 + (\psi(\x_1,\y_*')-\psi(\x_1,\y_*))\\
&\geq  -(\|\g\|+\|\h\|)\underbrace{\|\x_1\|}_{\leq R_1}\underbrace{\|\y_*'-\y_*\|}_{\leq\delta} - \underbrace{|\psi(\x_1,\y_*')-\psi(\x_1,\y_*)|}_{\leq\eps}\\
&\geq -(\|\g\|+\|\h\|)R_1\delta - \eps.
\end{align*}
Thus,
\begin{align*}
\Pro\left( \min_{\x\in\Sc_\x}\max_{\y\in\Sc_\y}\beta(\x,\y) \geq c + (\|\g\|+\|\h\|)R_1\delta + \eps \right) \leq \Pro\left( \min_{\x\in\Sc_\x^\delta}\max_{\y\in\Sc_\y^\delta}\beta(\x,\y) \geq c \right),
\end{align*}
and a further application of Lemma \ref{lem:conc1} shows that for all $t>0$:
\begin{align}\label{eq:vrika2}
\Pro\left( \min_{\x\in\Sc_\x}\max_{\y\in\Sc_\y}\beta(\x,\y) \geq c + (\sqrt{n}+\sqrt{m}+t)R_2\delta + \eps \right) - 2\exp(-t^2/4)\leq \Pro\left( \min_{\x\in\Sc_\x^\delta}\max_{\y\in\Sc_\y^\delta}\beta(\x,\y) \geq c \right),
\end{align}
Now, we can apply \eqref{eq:weknow} in order to combine \eqref{eq:vrika1} and \eqref{eq:vrika2} to yield the following:
\begin{align*}
\Pro\left( \min_{\x\in\Sc_\x}\max_{\y\in\Sc_\y}\alpha(\x,\y) \geq c  \right) \geq \Pro\left( \min_{\x\in\Sc_\x}\max_{\y\in\Sc_\y}\beta(\x,\y) \geq c + (\sqrt{n}+\sqrt{m}+1 + t)(R_1+R_2)\delta + 2\eps \right) - 4\exp(-t^2/4).
\end{align*}
This holds for all $\eps>0$ and all $t>0$. In particular, set $t=\delta^{-\frac{1}{2}}$ and take the limit of the right-hand side as $\eps\rightarrow 0$. Then, $t\rightarrow\infty$ and we can of course choose $\delta\rightarrow 0$, which proves that
\begin{align}\nn
\Pro\left( \min_{\x\in\Sc_\x}\max_{\y\in\Sc_\y}\alpha(\x,\y) \geq c  \right) \geq \Pro\left( \min_{\x\in\Sc_\x}\max_{\y\in\Sc_\y}\beta(\x,\y) > c  \right).
\end{align}
\end{proof}

\subsection{Auxiliary Results}

\begin{appdefn}[Lipschitz]
We say that a function $f:\R^d\rightarrow\R$ is Lipschitz with constant $L$ or is  $L$-Lipschitz if $|f(\x)-f(\y)|\leq L\|\x-\y\|$ for all $\x,\y\in\R^d$.
\end{appdefn}

\begin{apppropo}[Gaussian Lipschitz concentration]\label{prop:lip}(\cite[Theorem 5.6]{boucheron2013concentration})
Let $\x\in\R^d$ have i.i.d. $\Nn(0,1)$ entries and $f:\R^d\rightarrow\R$ be $L$-Lipschitz. Then, each one of the events $\{f(\x) > \E f(\x) + t \}$ and $\{f(\x) < \E f(\x) - t \}$ occurs with probability no greater than $\exp\left( -{t^2}/({2L^2}) \right)$.
\end{apppropo}


\begin{applem}\label{lem:conc1}
Let $\A\in\mathbb{R}^{{m}\times {n}}$, $g\in\R$, $\g\in\mathbb{R}^{m}$ and $\h\in\mathbb{R}^{n}$ have  entries i.i.d. $\Nn(0,1)$ and be independent of each other. Then, for all $t>0$, each one of the events
\begin{align}
\{ \|\A\|_2 + |g| \leq \sqrt{n}+\sqrt{m} + 1 + t \} \quad \text{ and } \quad \{ \|\h\|_2 + \|\g\|_2 \leq \sqrt{n}+\sqrt{m} + t \},
\end{align}
holds with probability at least $1 - 2\exp(-t^2/4)$.
\end{applem}

\begin{proof}
A well-known non-asymptotic bound on the largest singular value of an $m\times n$ Gaussian matrix shows (e.g. \cite[Corollary 5.35]{vershynin2010introduction}) that for all $t>0$:
\begin{align*}
\Pro\left(\|\A\|_2 > \sqrt{m} + \sqrt{n} + t\right)\leq \exp(-t^2/2).
\end{align*}
Also, $\|\cdot\|_2$ is an $1$-Lipschitz function and for a standard gaussian vector $\vb\in\R^d$: $\Exp\|\vb\|_2\leq\sqrt{d}$ . Applying Proposition \ref{prop:lip} we have that for all $t>0$ the events $\{|g|> 1+t\}$, $\{\|\g\|_2 > \sqrt{m} + t\}$ and $\{\|\h\|_2 > \sqrt{n} + t\}$, each one occurs with probability no larger than $\exp(-t^2/2)$. Combining those, 
\begin{align*}
\Pro\left( \|\A\|_2 + |g| \leq \sqrt{n}+\sqrt{m} + 1 + t \right) &\geq \Pro\left( \|\A\|_2  \leq \sqrt{n}+\sqrt{m}+ t/2~,~ |g| \leq 1 + t/2 \right) 
\\&\geq 1 - \Pro\left( \|\A\|_2  > \sqrt{n}+\sqrt{m}+ t/2\right) -\Pro\left(~ |g| > 1 + t/2 \right)
\\&\geq 1 - 2\exp(-t^2/4).
\end{align*}
The proof of the second statement of the lemma is identical and is omitted for brevity.
\end{proof}



\begin{applem}[Lipschitzness of Gordon's Optimization]\label{lem:lip}
Let $\Sc_1\subset\R^n$, $\Sc_\y\subset\R^m$ be compact sets and function $\Lc:\R^m\times\R^n\rightarrow\R$:
\begin{align*}
\Lc(\g,\h) &:= \min_{\x\in\Sc_1}\max_{\y\in\Sc_\y}~\|\x \|_2 \g^T\y  +  \|\y\|_2 \h^T\x + \f({\x,\y}).
\end{align*}
Further let $R_1 = \max_{\x\in\Sc_1}\|\x\|_2$ and $R_2 = \max_{\y\in\Sc_\y}\|\y\|_2$.
Then, $\Lc(\g,\h)$ is Lipschitz with constant $\sqrt{2}R_1R_2$.
\end{applem}
\begin{proof}
Fix any two pairs ($\g_1,\h_1$) and ($\g_2,\h_2$)  and let 
$$(\x_2,\y_2) = \arg\min_{\x\in\Sc_1}\max_{\y\in\Sc_\y}\|\x\|\g_2^T\y + \|\y\|\h_2^T\x + \psi(\x,\y),$$ 
and
$$\y_* = \arg\max_{\y\in\Sc_\y}\|\x_2\|\g_1^T\y + \|\y\|\h_1^T\x_2 + \psi(\x_2,\y).$$
Clearly,
$$ \Lc(\g_1,\h_1) \leq \|\x_2\|\g_1^T\y_* + \|\y_*\|\h_1^T\x_2 + \psi(\x_2,\y_*),$$
and
$$ \Lc(\g_2,\h_2) \geq \|\x_2\|\g_2^T\y_* + \|\y_*\|\h_2^T\x_2 + \psi(\x_2,\y_*),$$
Without loss of generality, assume $\Lc(\g_1,\h_1)\geq \Lc(\g_2,\h_2)$. Then,
\begin{align*}
\Lc(\g_1,\h_1) - \Lc(\g_2,\h_2) &\leq \|\x_2\|\g_1^T\y_* + \|\y_*\|\h_1^T\x_2 + \psi(\x_2,\y_*) - ( \|\x_2\|\g_2^T\y_* + \|\y_*\|\h_2^T\x_2 + \psi(\x_2,\y_*))\\
&\leq\|\x_2\|\y_*^T(\g_1 - \g_2) + \|\y_*\|\x_2^T(\h_1 - \h_2)\\
&\leq \sqrt{\|\x_2\|^2\|\y_*\|^2 + \|\y_*\|^2\|\x_2\|^2 }\sqrt{\|\g_1 - \g_2\|^2+\|\h_1 - \h_2\|^2}\\
&\leq R_1R_2\sqrt{2}\sqrt{\|\g_1 - \g_2\|^2+\|\h_1 - \h_2\|^2},
\end{align*}
where the penultimate inequality follows from Cauchy-Schwarz.
\end{proof}

\subsection{Proof of statement (iii) of Theorem  \ref{thm:main}}\label{sec:app_thm_proof}
\noindent\underline{Proof of \eqref{eq:norms}:} 
We start with introducing some notation that will simplify the exposition. In what follows, $\x$ is always constrained to belong to the set $\Sc_\x^\mn$; we simply write $\min_{\x}$ instead of $\min_{\x\in\Sc_\x^\mn}$. We will say that a sequence of events $\Ec^\mn$ holds/occurs with probability approaching (w.p.a.) 0 (or 1), if $\lim_{n\rightarrow\infty}\Pro(\Ec^\mn)=0$, (or 1).
Denote
$$
\ell(\eta) := \{ \alpha ~|~ |\alpha-\alpha_*| > \eta \}.
$$
We will prove that for all $\eta>0$, the event
 $\|\xPn\| \in \ell(\eta)$
 holds w.p.a. 1. 
 
 Consider the function $\Xfun:\Sc_\x^\mn\rightarrow\R$:
$$
\Xfunx = \max_{\y\in\Sc_\y^\mn} \y^T\G\x + \psi(\x,\y).
$$
Observe that $\Phi^\mn(\G)=\min_{\x}\Xfunx=\Xfun(\xPn;\G)$.
It is not hard to see that it suffices to prove that for all $\eta>0$ there exists $\delta:=\delta(\eta)>0$ such that 
\begin{align}\label{eq:goal}
\min_{\|\x\|\in\ell(\eta)}\Xfunx < \min_\x \Xfunx + \delta 
\end{align}
occurs w.p.a. 0.

In what follows, fix any $\eta>0$.
Proving \eqref{eq:goal} takes the following two steps: (i) upper bound $\min_{\x\in\Sc_\x^\mn}\Xfunx$, and (ii) lower bound $\min_{\|\x\|\in\ell(\eta)}\Xfunx$.

\emph{Step 1:} Fix some $\eps_1>0$ and consider the following event
\beq\label{eq:Ec}
\Ec^\mn(\eps_1) = \{   \min_{\x}\Xfunx > d_* + \eps_1  \}.
\eeq
Then, we may use statement (ii) of the theorem  (cf. \eqref{eq:dual}) to show that
\begin{align}\nn
\Pro(\Ec^\mn(\eps_1)) = \Pro( \Phi^\mn(\G) > d_* + \eps_1) \leq 2\Pro( \phi^\mn(\g,\h) \geq d_* + \eps_1 )
\end{align}
But, $\phi^\mn(\g,\h)\rP d_*$ by hypothesis of the theorem. Therefore,
$\Ec^\mn$ occurs w.p.a. 0.

\vspace{5pt}
\emph{Step 2:} 
Fix some $\eps_2>0$ and consider the following event:
\beq\label{eq:Hc}
\Hc(\eps_2):= \{ \min_{\|\x\|\in\ell(\eta)}\Xfunx < d_* + \eps_2 \}.
\eeq
Using statement (i) of the theorem (cf. \eqref{eq:dual}) we have
\begin{align}\label{eq:s1}
 \Pro(\Hc(\eps_2) ) \leq 2\Pro( \min_{\|\x\|\in\ell(\eta)}\xfunx \leq d_* + \eps_2 ).
\end{align}
We will upper bound the probability on the right hand side by conditioning on the event
$$
\{ \|\xpn\| \notin \ell(\eta/2) \},
$$
which occurs w.p.a. 1, by assumption. In this event, it is not hard to see that 
$$
\|\x\| \in \ell(\eta) \Rightarrow | \|\x\| - \|\xpn\| | \geq \eta/2.
$$
That is, conditioned on $\Ec^\mn$, the probability in \eqref{eq:s1} is further upper bounded by
\beq\label{eq:s2}
\Pro(~ \min_{| \|\x\| - \|\xpn\| | \geq \eta/2}\xfunx \leq d_* + \eps_2 ~).
\eeq
We will condition once more, only this time it will be on the event
$$
\{ \phi^\mn(\g,\h) \geq d_* - \eps_2/2 \},
$$
which occurs w.p.a. 1, by assumption. In this event, the probability in \eqref{eq:s2} is further upper bounded by
\beq\label{eq:s3}
\Pro(~ \min_{| \|\x\| - \|\xpn\| | \geq \eta/2}\xfunx \leq \phi^\mn(\g\,\h) + \eps_2/2 ~).
\eeq
Finally, we condition on the event
$$
\{ \xfunx \geq \phi^\mn(\g,\h) + \tau ( \|\x\| - \|\xpn\| )^2, ~\forall\x\in\Sc_\x \},
$$
which also occurs w.p.a. 1, by assumption. In this event, 
$$
 \min_{| \|\x\| - \|\xpn\| | \geq \eta/2}\xfunx \geq  \phi^\mn(\g,\h) + \tau(\eta/2)^2.
$$
Thus,
the probability in \eqref{eq:s2} is further upper bounded by
\beq\label{eq:s4}
\Pro(~ \tau(\eta/2)^2 \leq \eps_2/2 ~),
\eeq
which is of course a deterministic event. 
To sum up, following the chain of inequalities implied by \eqref{eq:s1}-\eqref{eq:s4}, we find that 
$$
\Pro( \Hc(\eps_2) ) \leq 2 \Pro(~ \tau(\eta/2)^2 \leq \eps_2/2 ~) + p^\mn(\eps_2),
$$
where $p^\mn(\eps_2)$ converges to $0$ as $n\rightarrow \infty$. In particular, $ \Hc(\eps_2)$ occurs w.p.a. 0, for all $\eps_2$ such that $ \eps_2 < 2\tau(\eta/2)^2$.

\vspace{5pt}
We are now ready to conclude the proof. For any $\eta>0$, choose $\eps_2:=\eps(\eta):=\tau(\eta/2)^2>0$, $\eps_1:=\eps_2/2$ and $\delta:=\eps_2/4>0$. Consider the events $\Ec(\eps_1)$ and $\Hc(\eps_2)$ as defined in \eqref{eq:Ec} and \eqref{eq:Hc}, respectively. For the particular choice of $\eps_1,\eps_2$ both events occur w.p.a. 0. Condition on both the complements of these events. Then, 
the probability of the event in \eqref{eq:goal} is upper bounded by 
$$
\Pro( d_* + \eps_2 < d_* + \eps_1 + \delta ) + p^\mn(\eps_1,\eps_2) = \Pro( 2 < 1  ) + p^\mn(\eps_1,\eps_2)= p^\mn(\eps_1,\eps_2),
$$
where $p^\mn(\eps_1,\eps_2)$ converges to 0 as $n\rightarrow\infty$.
This concludes the proof.


\subsection{An  alternative formulation of statement (iii) of Theorem \ref{thm:main}}\label{sec:recipe}

The main result of this section is Corollary \ref{cor:II}.
In order to set the stage for the statement of the corollary, we consider a simple example. Although somewhat trivial,  we will see later in Section \ref{sec:app}, that a slight modification to this example corresponds to the LASSO problem, cf. \eqref{eq:intro_LASSO}. Application of Corollary \ref{cor:II} to the analysis of the LASSO is detailed in Section \ref{sec:details}.

Let $\Sc^\mn_\x=\Kc^\mn\cap (K \Bc^n)$, $\Sc^\mn_\y=\Bc^m$ and $\psi^\mn(\x,\y) =0$, where $\Kc^\mn$ is a sequence of convex closed cones in $\R^n$, $\Bc^k$ denotes the unit $\ell_2$-ball in $\R^k$ and $K>0$ is a fixed constant. 
 In this case, $\xfunx$ simplifies as follows
\begin{align}
\xfunx  &= \max_{\|\y\|_2\leq 1} \|\x\|_2\g^T\y + \|\y\|_2\h^T\x -\z^T\y = \max_{0\leq \beta \leq 1} \beta \|\x\|_2\|\g\|_2 + \beta\h^T\x ,\label{eq:above99}
\end{align}
where we have used  $\max_{\|\y\|_2=\beta}\g^T\y = \beta\|\g\|_2$ to reduce the maximization into one that involves only scalars. What is more, the objective function in \eqref{eq:above99} is now convex in $\x$ and linear (thus, concave) in $\beta$.  This, and the fact that both the constraint sets are convex and compact permits flipping the order of min-max in \eqref{eq:2a} \cite[Cor.~37.3.2]{Roc70}:
\begin{align}
\phi^\mn(\g,\h) =  \min_{\x\in\Sc^\mn_\x}\max_{0\leq \beta \leq 1} \beta \|\x\|_2\|\g\|_2 + \beta\h^T\x  &= \max_{0\leq \beta \leq 1}\min_{\x\in\Sc^\mn_\x} \beta \|\x\|_2\|\g\|_2 + \beta\h^T\x , \nn
\\
&= \max_{0\leq \beta \leq 1} \min_{0\leq\alpha\leq K} \underbrace{ \min_{\substack{\x\in\Kc^\mn \\ \|\x\|_2=\alpha}} \beta \|\x\|_2\|\g\|_2 + \beta\h^T\x }_{:=\phi^\mn_\rn(\alpha,\beta;\g,\h) }.\label{eq:above100}
\end{align}
The function $\phi^\mn_\rn(\alpha,\bb;\g,\h)$ as defined in \eqref{eq:above100} is easy to evaluate, using the homogeneity of the cone $\Kc^\mn$:
$$
\phi^\mn_\rn(\alpha,\beta;\g,\h)  = \alpha\beta\|\g\|_2 - \alpha\beta\max_{\substack{\x\in\Kc^\mn \\ \|\x\|_2=1}}\h^T\x, 
$$
This function is linear (thus, convex) in $\alpha$ and is of course concave in $\beta$. Hence, we can flip the order of min-max in \eqref{eq:above100} to conclude with
\begin{align}
\phi^\mn(\g,\h) =  \min_{0\leq\alpha\leq K} \max_{0\leq \beta \leq 1}\{  \alpha\beta\|\g\|_2 - \alpha\beta\max_{\substack{\x\in\Kc^\mn \\ \|\x\|_2=1}}\h^T\x\},\label{eq:above101}
\end{align}
\eqref{eq:above101} characterizes gordon's optimization in \eqref{eq:2a} as a rather simple scalar min-max problem. Furthermore, it is seen that any minimizer $\alpha_*$ of \eqref{eq:above101} corresponds to the $\ell_2$-norm $\|\xp\|_2$ of some optimal $\xp$ in \eqref{eq:2a}. Hence, conditions (a) and (b) of statement (iii) in Theorem \ref{thm:main} can be worked out via an asymptotic probabilistic analysis of \eqref{eq:above101}. Furthermore, to check condition 3) observe that for any feasible $\tilde\x\in\Sc_\x$ with $\|\tilde\x\|_2=\tilde\alpha$, we have
\begin{align}
\upsilon^\mn(\tilde\x;\g,\h) \geq \min_{\x\in\Sc_\x, \|\x\|_2=\tilde\alpha} \xfunx &= \min_{\x\in\Sc_\x, \|\x\|_2=\tilde\alpha}  \max_{0\leq \beta \leq 1} \beta \|\x\|_2\|\g\|_2 + \beta\h^T\x \label{eq:above102}\\
&\geq  \max_{0\leq \beta \leq 1} \min_{\x\in\Sc_\x, \|\x\|_2=\tilde\alpha}  \beta \|\x\|_2\|\g\|_2 + \beta\h^T\x \label{eq:above103}\\
&=  \max_{0\leq \beta \leq 1} \phi^\mn_\rn(\tilde\alpha,\beta;\g,\h).\label{eq:above104}
\end{align} 
In \eqref{eq:above102} we used \eqref{eq:above99}; for the min-max inequality in \eqref{eq:above103} see for example \cite[Lem.~36.1]{Roc70}; for \eqref{eq:above104} recall \eqref{eq:above100}. Combining this with \eqref{eq:above101}, it can be seen that condition (c) of the theorem corresponds to some sort of strong-convexity  of the function $\max_{0\leq \beta \leq 1} \phi^\mn_\rn(\cdot,\beta;\g,\h)$. 
Corollary \ref{cor:II} formalizes these ideas. 
See Section \ref{sec:app} and references therein for specific examples where Corollary \ref{cor:II} applies.

\begin{cor}\label{cor:II}
Let the same setup as in Theorem \ref{thm:main} and the assumptions of statement (ii) therein hold. Suppose the following are true:
\begin{enumerate}[(1)]
\item We can express $\xfunx$ in \eqref{eq:xfun} as $\max_{\bb\in\Sc_\bb^\mn} \nu^\mn(\x,\bb;\g,\h)$, for $\Sc_\bb^\mn$ convex, compact, and, $\nu^\mn(\x,\bb)$ continuous and convex-concave.
\item $K^\mn:=\max_{\x\in\Sc_\x^{(n)}}\|\x\|=K<\infty$
and the function $\phi^\mn_\rn(\cdot,\bb;\g,\h):[0,K]\rightarrow\R$ defined as
$
\phi^\mn_\rn(\alpha,\bb;\g,\h) := \min_{\substack{\x\in\Sc^\mn_\x, \|\x\|=\alpha}}\nu^\mn(\x,\bb;\g,\h)$
is convex with probability one.
\item There exists   $d:[0,K]\rightarrow\R$  such that $\max_{\bb\in\Sc^\mn_\bb}\phi^\mn_\rn(\alpha,\bb;\g,\h)\rP d(\alpha)$ for all $\alpha\in[0,K]$.
\item $d(\cdot)$ is strongly convex and has unique minimizer $\alpha_*$ in $[0,K]$.
\end{enumerate}
Then, $ \Phi^\mn(\G) \rP d(\alpha_*)$ and $\|\x_\Phi^\mn(\G)\|  \rP \alpha_*.$
\end{cor}


%

\noindent A few remarks are in place.

\begin{enumerate}
\item Essentially assumptions (1) and (2) of the corollary are equivalent to being able to express \eqref{eq:2a} as:
\beq
\phi^\mn(\g,\h) = \min_{0\leq\alpha\leq K} \max_{\bb\in\Sc_\bb} \phi^\mn_{\|\cdot\|}(\alpha,\bb;\g,\h),
\eeq
where the optimal value of $\alpha$ corresponds to $\|\xpn\|$. Please see \eqref{eq:above101}. This step reduces the minimization in \eqref{eq:2a} to a scalar minimization and we often refer to it as ``scalarization".

\item Assumption (2)  essentially restricts $\{\x^\mn_\Phi(\G)\}_{n\in\mathbb{N}}$ to be bounded (in probability). Similarly, (4) presumes that for any $\alpha\in[0,K]$, the sequence $\{\max_{\bb\in\Sc^\mn_\bb}\phi_{\|\cdot\|}^\mn(\alpha,\bb;\g,\h)\}_{n\in\mathbb{N}}$ is stochastically bounded\footnote{a sequence $\{\mathcal{X}^\mn\}_\mn$ of random variables is stochastically bounded if for any $\eps$ there exists $C>0$ large enough such that $\lim_{n\rightarrow\infty}\Pro(|\mathcal{X}_n|>C)\leq\eps$.}. Both these restrictions are primarily imposed to to simplify the statement of the theorem and can in principle  be satisfied after proper normalization with $n$.
\item To find a candidate function $d(\cdot)$ that satisfies (4), one needs to  calculate the probability limit of \\ $\max_{\bb\in\Sc_\bb^\mn}\phi^\mn_\rn(\alpha,\bb;\g,\h)$ for each $\alpha$. 
Thanks to Lipschitzness (see Lemma \ref{lem:lip}) the involved random quantities concentrate. Thus, the cost of the deterministic optimization $\min_\alpha d(\alpha)$ reflects the asymptotic behavior of the original random problem. 
\end{enumerate}
In what follows we prove Corollary \ref{cor:II}.
\begin{proof}
It is convenient to define
$$
\hat\phi_\rn^\mn(\alpha;\g,\h) := \max_{\bb\in\Sc^\mn_\bb}\phi^\mn_\rn(\alpha,\bb;\g,\h).
$$
Let all the assumptions of the corollary hold true. We will show that these imply the three conditions of statement (iii) of Theorem \ref{thm:main}. This will suffice to prove the result.

First, we claim that 
\begin{align}\label{eq:claim_1}
\phi^\mn(\g,\h) = \min_{0\leq\alpha\leq K} \hat\phi_\rn^\mn(\alpha;\g,\h).
\end{align}
This follows from the following sequence of equations,
\begin{align}
\phi^\mn(\g,\h) = \min_{\x\in\Sc_\x} \xfunx &= \min_{\x\in\Sc_\x} \max_{\bb\in\Sc_\bb} \nu^\mn(\x,\bb;\g,\h) \nn \\
&= \max_{\bb\in\Sc_\bb} \min_{\x\in\Sc_\x}  \nu^\mn(\x,\bb;\g,\h) \label{eq:this_1}\\
&= \max_{\bb\in\Sc_\bb} \min_{0\leq \alpha\leq K} \min_{\substack{\x\in\Sc_\x\\ \|\x\|=\alpha}}  \nu^\mn(\x,\bb;\g,\h)  = \max_{\bb\in\Sc_\bb} \min_{0\leq \alpha\leq K}  \phi^\mn(\alpha,\bb;\g,\h) \nn
\\ &= \min_{0\leq \alpha\leq K} \max_{\bb\in\Sc_\bb} \phi^\mn(\alpha,\bb;\g,\h)  = \min_{0\leq\alpha\leq K} \hat\phi_\rn^\mn(\alpha;\g,\h). \label{eq:this_2}
\end{align}
\eqref{eq:this_1} follows from \cite[Cor.~37.3.2]{Roc70} since $\nu^\mn(\x,\bb;\g,\h)$ is convex-concave and the constraint sets are convex and compact (cf. assumption (1)). The first equality in \eqref{eq:this_2} follows from another application of \cite[Cor.~37.3.2]{Roc70}: the constraint sets are clearly convex and compact; the objective function is convex in $\alpha$ by assumption (2) and is concave in $\bb$ as the point-wise minimum of concave functions. 

Next, consider any feasible $\tilde\x\in\Sc_\x$ with $\|\tilde\x\|=\tilde\alpha$. We prove that
\begin{align}\label{eq:claim_2}
\upsilon^\mn(\tilde\x;\g,\h) \geq \hat\phi_\rn^\mn(\tilde\alpha;\g,\h).
\end{align}
For this, we have the following chain of inequalities,
\begin{align}
\upsilon^\mn(\tilde\x;\g,\h) \geq \min_{\substack{\x\in\Sc_\x \\ \|\x\|=\tilde\alpha}} \upsilon^\mn(\tilde\x;\g,\h) &= \min_{\substack{\x\in\Sc_\x \\ \|\x\|=\tilde\alpha}} \max_{\bb\in\Sc_\bb} \nu^\mn(\x,\bb;\g,\h) \nn\\
&\geq  \max_{\bb\in\Sc_\bb} \min_{\substack{\x\in\Sc_\x \\ \|\x\|=\tilde\alpha}} \nu^\mn(\x,\bb;\g,\h)  \label{eq:this_3} \\
&=  \max_{\bb\in\Sc_\bb}  \phi^\mn(\tilde\alpha,\bb;\g,\h) = \hat\phi_\rn^\mn(\tilde\alpha;\g,\h).\nn
\end{align}
\eqref{eq:this_3} follows from the min-max inequality \cite[Lem.~36.1]{Roc70}.

Finally, we show that $\hat\phi_\rn^\mn(\alpha;\g,\h)$ converges \emph{uniformly} in probability to $d(\alpha)$ in the compact set $[0,K]$, i.e. for all $\delta>0$ the following holds with probability approaching one in the limit $n\rightarrow\infty$ :
\begin{align}\label{eq:claim_3}
\sup_{\alpha\in[0,K]}  \left| { \hat\phi_\rn^\mn(\alpha;\g,\h) } - {d(\alpha)}  \right| <\delta .
\end{align}
Assumption (3) of the corollary guaranties point-wise convergence, i.e. for all $\alpha\in[0,K]$ $ \hat\phi_\rn^\mn(\alpha;\g,\h)\rP d(\alpha)$. Notice that $\hat\phi_\rn^\mn(\cdot;\g,\h)$ is convex as the point-wise maximum of convex functions (cf. assumption (2)). This is critical to establish \eqref{eq:claim_3}, since ``for \emph{convex} functions pointwise convergence in probability  implies uniform convergence in compact subsets" \cite[Corollary II.1]{AG1982} (see also \cite[Theorem 2.7]{NF36}). 
%

Now, we are ready to show the three conditions of statement (iii) in Theorem \ref{thm:main}, namely
\begin{align}\label{eq:cond_1}
\phi^\mn(\g,\h) \rP d(\alpha_*),
\end{align}
\begin{align}\label{eq:cond_2}
\|\x_\phi^\mn(\g,\h)\|\rP\alpha_*,
\end{align}
\begin{align}\label{eq:cond_3}
\xfunx \geq \phi^\mn(\g,\h) + \tau(\|\x\|-\|\xpn\|), \forall \x\in\Sc_\x, \text{ with probability 1 } .
\end{align}
Here, recall that $\alpha_*$ is the unique minimizer of $d$ from assumption (4) of the corollary. Also, we choose $\tau>0$ be such that 
\begin{align}\label{eq:sc}
d(\alpha) \geq d(\az) + \tau|\alpha-\az|^2,~~ \forall \alpha\in [0,K].
\end{align}
This is also guarantied by the strong convexity assumption (4).

Let us start with \eqref{eq:cond_1}. From \eqref{eq:claim_3}, $\min_{0\leq\alpha\leq K}\hat\phi_\rn^\mn(\alpha;\g,\h)\rP d(\alpha_*)$. Then, the claim follows from \eqref{eq:claim_1}.

\eqref{eq:cond_1} follows with a similar argument. From the equalities in \eqref{eq:this_2}, $$\|\xpn\|=\alpha_*(\g,\h):=\arg\min_{0\leq\alpha\leq K}\hat\phi_\rn^\mn(\alpha;\g,\h).$$
The uniform convergence result in \eqref{eq:claim_3} and the fact that $\alpha_*$ is unique, can be used to show that $\alpha_*(\g,\h)\rP\alpha_*$; for example see \cite[Thm.~2.1]{NF36}.

Finally, we prove \eqref{eq:cond_3} as follows. 
%
%
Let any $\tilde\x\in\Sc_\x$ with $\|\tilde\x\|=\tilde\alpha$. Fix $\eps>0$ and let $\delta_1, \delta_2>0$ be constants to be determined later in the proof.  From  uniform convergence in \eqref{eq:claim_2}, the following events occurs w.p.a. 1,

\begin{subequations}
\beq\label{eq:one}
\hat\phi^\mn_\rn(\tilde\alpha;\g,\h) \geq  d(\tilde \alpha) - \delta_1,
\eeq
\beq\label{eq:two}
\phi^\mn(\g,\h) \leq \hat\phi^\mn_\rn(\alpha_*;\g,\h) \leq d(\alpha_*)+\delta_1.
\eeq
\end{subequations}
Furthermore, from \eqref{eq:cond_2} w.p.a. 1,
\beq\label{eq:three}
|\|\xpn\|-\alpha_*| \leq \delta_2.
\eeq
%
Also the strong convexity assumption (4) (cf. \eqref{eq:sc}) shows that
\begin{align}\label{eq:three}
d(\tilde\alpha) \geq d(\az) +  \tau(\tilde\alpha-\alpha_*)^2.
\end{align}
Combining \eqref{eq:claim_2} with \eqref{eq:one}, \eqref{eq:two} and \eqref{eq:three} we find that w.p.a. 1,
\begin{align}\label{eq:almost}
\upsilon^\mn(\tilde\x;\g,\h) \geq  \phi^\mn(\g,\h) -2\delta_1 + \tau(\|\tilde\x\|-\alpha_*)^2.
\end{align}
By the triangular inequality and \eqref{eq:three}, 
$$|\|\tilde\x\|-\alpha_*| \geq |\|\tilde\x\|-\|\xpn\|| - |\|\xpn\|-\alpha_*| \geq |\|\tilde\x\|-\|\xpn\|| - \delta_2.$$
Combine this with \eqref{eq:almost} to yield
\begin{align}\label{eq:almost}
\xfunx \geq  \phi^\mn(\g,\h)  + \tau(\|\tilde\x\|-\|\xpn\|)^2 + \tau\delta_2^2-2\delta_1.
\end{align}
Choose $\delta_1$ and $\delta_2$ such that $\tau\delta_2^2\geq2 \delta_1$  to conclude with the desired as in \eqref{eq:cond_3}.
This concludes the proof of the corollary.
\end{proof}

\subsection{On the compactness and continuity assumptions of Theorem \ref{thm:main}}\label{sec:tech}

First, we argue that the set of minima in \eqref{eq:gL1}, say $\Ws$, is nonempty and compact. Using the assumptions on $\loss(\cdot)$ and $f(\cdot)$, the objective function $\loss(\A\w-\z) + \la f(\x_0+\w)$ is continuous on $\R^n$ and coercive for any $\la\geq 0$.  Hence, the claim follows from Weiestrass' Theorem as in \cite[Prop. 2.1.1]{Bertsekas}. This implies that the optimal cost and the set of minima in \eqref{eq:gL1} does not change if we restrict the minimization over the set $\Sc_\w:=\{\w\in\R^n | \|\w\|_2\leq K_\w\}$ provided
\begin{align}\label{eq:cw1}
\max_{\w\in\Ws}\|\w\|_2\leq K_\w< \infty.
\end{align}

Next, let us argue on the compactness of the set of maxima in \eqref{eq:gL2}.  From convexity and continuity of $\loss$ it can be seen (e.g. \cite[Thm. 23.5]{Roc70}) that $\ub^T\vb-\loss^*(\ub)$ achieves its supremum over $\ub$ at a point in the subdifferential of $\loss(\cdot)$ at $\vb$, i.e.,  $\ub\in\partial\loss(\vb)$. Translating this into \eqref{eq:gL2}: for any $\w$, $\arg\sup_{\ub}\{\ub^T(\A\w-\z)-\loss^*(\ub)\}\in\partial\loss(\A\w-\z)$. It is then clear that for $\w\in\Sc_\w$, the set of $\ub$ at which the supremum is achieved in \eqref{eq:gL2} can be expressed as $\mathcal{U}_*:=\bigcup\{  \partial\loss(\A\w-\z) | \w\in\Sc_\w \}$. By assumption, $\loss(\cdot)$ is continuous and convex. Also, $\Sc_\w$ is a bounded subset of $\R^n$. Thus, $\mathcal{U}_*$ is bounded \cite[Prop. 4.2.3]{Bertsekas}.
This implies that the optimal cost and the set of minima in \eqref{eq:gL2} does not change if we replace the $\sup$ operator with a maximization over the set $\Sc_\ub:=\{\ub\in\R^m | \|\ub\|_2\leq K_\ub\}$ provided
\begin{align}\label{eq:cu1}
\max_{\ub\in\mathcal{U}_*}\|\ub\|_2\leq K_\ub< \infty.
\end{align}
We remark that if $\dom\loss^*$ is bounded, we can of course choose $\Sc_\ub=\dom\loss^*$. For instance, if $\loss(\cdot)=\|\cdot\|$, then $\dom\loss^*=\{\ub ~|~ \|\ub\|_*\leq 1\}$.
To summarize, \eqref{eq:gL2} is equivalent to \eqref{eq:gL3} provided that $K_\w$ and $K_\ub$ satisfy \eqref{eq:cw1} and \eqref{eq:cu1}.

Finally, we comment on the continuity of the function $\psi(\w,\ub):=-\ub^T\z-\loss^*(\ub)+\la f(\x_0+\w)$.  Application of Theorem \ref{thm:main} requires $\psi(\w,\ub)$ to be continuous on $\Sc_\w\times\Sc_\ub$. By assumption, $f(\cdot)$ is continuous. Thus, the only component we need to worry about is $\loss^*(\cdot)$.
Recall that $\loss^*(\cdot)$ was defined in Section \ref{sec:app_gen} as an extended real-valued function. Consequently, $\f(\w,\ub)$ is also an extended real-valued function and $\dom\psi = \Sc_\w\times\dom\loss^*$. It can be shown that Theorem \ref{thm:main} continues to hold provided that $\f(\w,\ub)$ is continuous on its effective domain. But this is true by the assumption made in Section \ref{sec:app_gen} on $\loss^*$ being continuous on $\dom\loss^*$.

%

A subtlety in the discussion thus far is that the constraints imposed on $K_\w$ and $K_\ub$ in \eqref{eq:cw1} and \eqref{eq:cu1} depend on the specific values of $\A\in\R^{m\times n}$, $\z\in\R^m$ and $\x_0\in\R^n$. What we have shown is that for \emph{any} triplet $\A,\z,\x_0$ there exist $K_\w(\A,\z,\x_0)<\infty$ and $K_\ub(\A,\z,\x_0)<\infty$ such that \eqref{eq:gL2} is equivalent to \eqref{eq:gL3}. Choosing the maximum such constants over all triplets guarantees the desired conclusion.

\end{document}